\documentclass[11pt, reqno]{article}

\usepackage{jheppub}
\usepackage{amssymb}
\usepackage{amsmath}
\usepackage{amsthm}
\usepackage[usenames,dvipsnames]{xcolor}
\usepackage{epsfig}
\usepackage{dcolumn}
\usepackage{upgreek}
\usepackage{setspace}
\usepackage{enumitem}
\usepackage{array,multirow,bigdelim,arydshln}
\usepackage{appendix}
\usepackage{xparse}
\usepackage{stmaryrd}
\usepackage{graphicx}
\usepackage[utf8]{inputenc}
\usepackage{shuffle}
\hypersetup{
	colorlinks,
	urlcolor=Maroon,
	linkcolor=Maroon,
	citecolor=Maroon
}

\usepackage{chngcntr}
\counterwithin{figure}{section}

\NewDocumentCommand{\binomial}{omm}
 {%
  \genfrac(){0pt}{}{#2}{#3}%
  \IfValueT{#1}{_{\!#1}}%
 }
\NewDocumentCommand{\eulerian}{omm}
 {%
  \genfrac<>{0pt}{}{#2}{#3}%
  \IfValueT{#1}{_{\!#1}}%
 }

\def \s {\sigma}

\usepackage{latexsym}
\usepackage{tikz}

\allowdisplaybreaks
\graphicspath{{figures/}}

\title{$\Delta$-Algebra and Scattering Amplitudes}

\author[a]{Freddy Cachazo,}\emailAdd{fcachazo@pitp.ca}
\author[b]{Nick Early,}\emailAdd{earlnick@gmail.com}
\author[a,c,d]{Alfredo Guevara,}\emailAdd{aguevara@pitp.ca}
\author[a,c]{and Sebastian Mizera}\emailAdd{smizera@pitp.ca}

\affiliation[a]{Perimeter Institute for Theoretical Physics, Waterloo, ON N2L 2Y5, Canada}
\affiliation[b]{Massachusetts Institute of Technology, Cambridge, MA, United States}
\affiliation[c]{Department of Physics \& Astronomy, University of Waterloo, Waterloo, ON N2L 3G1, Canada}
\affiliation[d]{CECs Valdivia \& Departamento de F\'isica, Universidad de Concepci\'on, Casilla 160-C,\\ Concepci\'on, Chile}

\abstract{In this paper we study an algebra that naturally combines two familiar operations in scattering amplitudes: computations of volumes of polytopes using triangulations and constructions of canonical forms from products of smaller ones. We mainly concentrate on the case of $G(2,n)$ as it controls both general MHV leading singularities and CHY integrands for a variety of theories. This commutative algebra has also appeared in the study of configuration spaces and we called it the $\Delta$-algebra. As a natural application, we generalize the well-known square move. This allows us to generate infinite families of new moves between non-planar on-shell diagrams. We call them {\it sphere moves}. Using the $\Delta$-algebra we derive familiar results, such as the KK and BCJ relations, and prove novel formulas for higher-order relations. Finally, we comment on generalizations to $G(k,n)$.}

\begin{document}

\maketitle
\addtocontents{toc}{\protect\setcounter{tocdepth}{1}}
\def \tr {\nonumber\\}
\def \la  {\langle}
\def \ra {\rangle}
\def\hset{\texttt{h}}
\def\gset{\texttt{g}}
\def\sset{\texttt{s}}
\def \be {\begin{equation}}
\def \ee {\end{equation}}
\def \ba {\begin{eqnarray}}
\def \ea {\end{eqnarray}}
\def \k {\kappa}
\def \h {\hbar}
\def \r {\rho}
\def \l {\lambda}
\def \be {\begin{equation}}
\def \en {\end{equation}}
\def \bes {\begin{eqnarray}}
\def \ens {\end{eqnarray}}
\def \red {\color{Maroon}}
\def \pt {{\rm PT}}
\def \s {\sigma}
\def \ls {{\rm LS}}
\def \ma {\Upsilon}

\numberwithin{equation}{section}
\setcounter{page}{2}

\theoremstyle{plain}
\newtheorem{lemma}{Lemma}[section]
\newtheorem{theorem}{Theorem}[section]

\theoremstyle{definition}
\newtheorem{definition}{Definition}[section]
\newtheorem{remark}{Remark}[section]
\newtheorem{corollary}{Corollary}[section]
\newtheorem{example}{Example}[section]
\newtheorem{proposition}{Proposition}[section]
\newtheorem{conjecture}{Conjecture}[section]

\section{Introduction: Motivating the $\Delta$-Algebra}

Scattering amplitudes in ${\cal N}{=}4$ super Yang--Mills in four dimensions are known to be constructible in terms of two classes of three-particle amplitudes, usually represented as black and white trivalent vertices \cite{ArkaniHamed:2012nw}. Gluing such vertices so that each edge is on-shell gives rise to on-shell diagrams. Choosing an ordering for $n$ external particles, planar on-shell diagrams become plabic graphs \cite{postnikov2006total} and are deeply connected to the positive Grassmannians $G_{\geq 0}(k,n)$ \cite{ArkaniHamed:2012nw,postnikov2006total}. In physics, $k$ represents the N$^{k-2}$MHV sector the diagram belongs to.
It is known that the simplest sector, i.e., $k=2$ or MHV, is special in many ways. In particular, any on-shell diagram, planar or not, associated with a top dimensional region of $G(2,n)$ can be characterized by a set of $n{-}2$ triples of labels \cite{Arkani-Hamed:2014bca}. This is because each such on-shell diagram contains exactly $n{-}2$ black trivalent vertices even after contracting like-colored vertices to make the diagram  bipartite. Moreover, in the bipartite diagram each white vertex has exactly one external leg attached to it and each black vertex is connected to three distinct white vertices. As explained in \cite{Arkani-Hamed:2014bca}, and recently explored in \cite{He:2018okq} from a novel view point, one can associate a 2-form to each black vertex
\be\label{wis}
\Omega_{abc} = d\log \frac{\langle a,b\rangle}{\langle a,c\rangle}\wedge d\log \frac{\langle b,c\rangle}{\langle a,c\rangle}
\ee
so that the top form, i.e., a $2(n{-}2)$-form associated with a given on-shell graph is simply given by
\be\label{wedges}
\Omega_{\cal T} = \bigwedge_{\tau\in {\cal T}}\Omega_{\tau_1,\tau_2,\tau_3},
\ee
where ${\cal T}$ is a list of $n-2$ triples $\tau$ defining the diagram.

Also in \cite{Arkani-Hamed:2014bca}, an alternative formula for a rational function associated with the on-shell diagram ${\cal T}$ was given in term of the reduced determinant of a matrix $M$ constructed from ${\cal T}$
\be\label{rtnl}
F_{\cal T} = 
\frac{\left({\rm det}'M\right)^2}{\prod_{\tau\in {\cal T}}\langle\tau_1,\tau_2\rangle\langle\tau_2,\tau_3\rangle\langle\tau_3,\tau_1\rangle}.
\ee
Here we rewrite the determinants in \eqref{rtnl} in terms of integrals over Grassmann variables $\theta_a, \chi_a$ motivating the definition of the following object:
\be
\Delta_{abc}:=\frac{\left(\theta_{a}\langle b,c\rangle+\theta_{b}\langle c,a\rangle+ \theta_{c}\langle a,b\rangle \right)\left(\chi_{a}\langle b,c\rangle+\chi_{b}\langle c,a\rangle+\chi_{c}\langle a,b\rangle \right)}{\langle a,b\rangle\langle b,c\rangle\langle c,a\rangle}.
\ee
This object naturally maps to $\Omega_{abc}$ under the map introduced by He and Zhang \cite{He:2018okq} that takes Grassmann variables to differential forms.  

The formula for $F_{\cal T}$ can then be written as Grassmann integrations over the product of $n{-}2$ $\Delta$'s; one for each triple of labels in ${\cal T}$. Moreover, we identify the {\it integrand} as the physically relevant object by making the following definition
\be\label{DLS}
{\rm LS}_{\cal T}:= \prod_{\tau\in {\cal T}}\Delta_{\tau_1,\tau_2,\tau_3}.
\ee
${\rm LS}$ stands for ``leading singularity" which is the terminology for the physical meaning of the quantity \cite{Britto:2004nc,Buchbinder:2005wp,Cachazo:2008vp}. Note that ${\rm LS}_{\cal T}$ has the same structural form as \eqref{wedges}. 

Two crucial properties of $\Delta$'s are that they are nilpotent, $\Delta^2=0$, and they commute, $\Delta\Delta'=\Delta'\Delta$. Using these two properties, $\text{LS}_{\cal T}$ can be rewritten as
\be
{\rm LS}_{\cal T} = \frac{1}{(n-2)!}\left(\sum_{\tau\in {\cal T}}\Delta_{\tau_1,\tau_2,\tau_3}\right)^{n-2}.
\ee
The simplest leading singularity is known as the Parke--Taylor function and is obtained by choosing, e.g., $\mathcal{T}=\{(123),(134),\ldots (1,n{-}1,n)\}$ and therefore
\be
{\rm LS}_{\text{Parke--Taylor}} = \frac{1}{(n-2)!}\left(\sum_{i=2}^{n-1}\Delta_{1,i,i+1}\right)^{n-2}.
\ee

In a different line of developments, the study of positive geometries \cite{Arkani-Hamed:2017tmz}, and the amplituhedron \cite{Arkani-Hamed:2013jha} has led to the development of volume functions usually denoted by $[i_1\ldots i_{m+1}]$, see, e.g., \cite{Arkani-Hamed:2017vfh,Ferro:2018vpf}. Here $m$ denotes the $\mathbb{CP}^{m}$ where the object lives. These functions are associated with simplices that can be put together to form more complicated geometries. The corresponding volume is the sum of the volume functions. Letting $m=2$ gives rise to volume (or area) functions with three labels. The simplest example is an $n$-sided polygon whose area is computed, e.g., as
\be
\sum_{i=2}^{n-1}[1,i,i+1].
\ee

These two lines of development naturally motivate the study of the algebra generated by $\Delta$'s which we call the $\Delta$-algebra.

\begin{figure}
\centering\includegraphics[scale=.9]{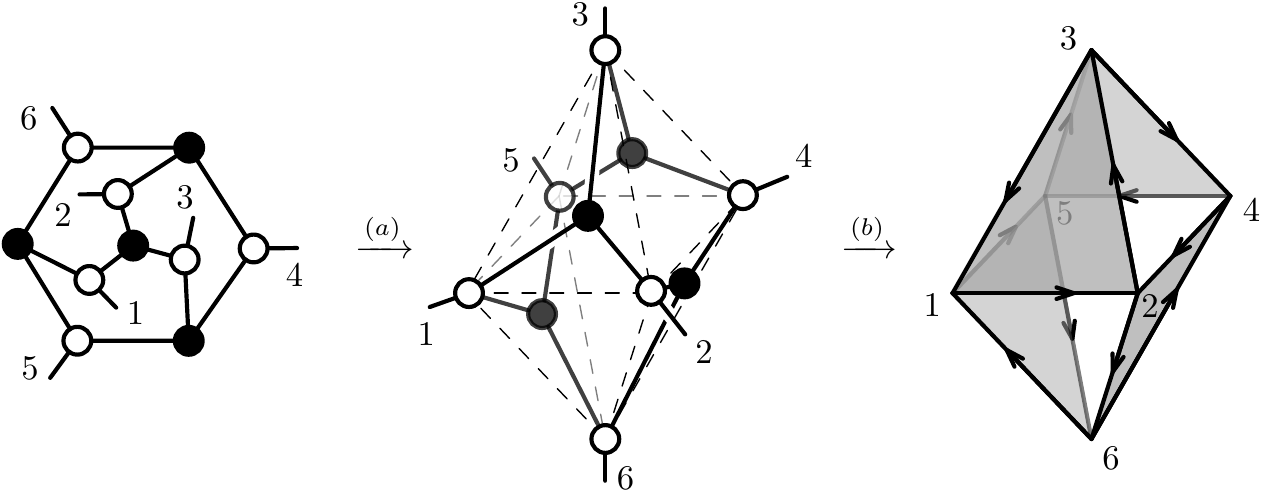}
\caption{\label{fig:intro1}Translation from an on-shell diagram to a set of triangles living on a sphere, for the example of the octahedral leading singularity associated to the set of triples $\mathcal{T}=\{(1,2,3),(3,4,5),(5,6,1),(6,4,2)\}$. (a) Place black and white vertices at the faces and vertices of the octahedron respectively. (b) Associate a shaded face to each black vertex. Orientation of each face is that of the corresponding triple.}
\end{figure}

Interestingly, the $\Delta$-algebra has also appeared in the study of configuration spaces, in particular the configuration space of $n$ distinct points in $SU(2)$. More precisely, in \cite{early2018configuration}, the cohomology ring of the configuration space, $\text{H}^\ast(\text{Conf}_n(SU(2))\slash SU(2), \mathbb{C})$, of $n$ distinct points on $SU(2)$ modulo the diagonal action, was constructed as a subalgebra, denoted $\mathcal{V}^n$, of the cohomology ring $\text{H}^\ast(\text{Conf}_n(\mathbb{R}^3), \mathbb{C})$.\footnote{See also \cite{moseley2017orlik} for a related conjectural description of the cohomology ring  $\text{H}^\ast(\text{Conf}_n(SU(2))\slash SU(2), \mathbb{C})$.  
}
By comparing generators and relations, one finds that in fact the generators $\Delta_{abc}$ satisfy the same relations as the generators $v_{abc}$ for $\mathcal{V}^n$.  Therefore the $\Delta$-algebra is a representation of $\mathcal{V}^n$, with $v_{abc}\mapsto \Delta_{abc}$. For background on configuration spaces, see, e.g., the review \cite{knudsen2018configuration} and the classic papers \cite{Arnol'd1969,10.2307/2946631,TOTARO19961057,kriz1994rational}. For related motivating work on permutohedral tessellations and blades, see \cite{OcneanuLectures,Early:2018mac}.

In this paper we initiate the study of the $\Delta$-algebra applications in scattering amplitudes. In particular, MHV on-shell diagrams that produce leading singularities. A crucial property of on-shell diagrams is that any two diagrams related by an operation known as the ``square move" give rise to the same rational function \eqref{rtnl}, i.e., the same physical object. Moreover, if the diagrams are planar (and reduced) it was proven by Postnikov \cite{postnikov2006total} that the square move, combined with expansion and contraction moves of like-colored vertices, is enough to define equivalence classes of diagrams, known as plabic graphs, each encoding a different cell in $G_{\geq 0}(k,n)$.

In the non-planar case, already at six points this is not the case as found by the second author in the study of permutohedral tessellations and blades \cite{Early:2018mac}. There are two MHV on-shell diagrams not connected via square moves that give rise to the same rational function via something that can be called the octahedral move. As anticipated in \cite{Early:2017lku}, finding a structure such as the $\Delta$-algebra associated with on-shell diagrams makes these new identifications natural and easy to prove. In fact, we show that the square and octahedral moves are the simplest cases of an infinite family of moves we call \emph{sphere moves}. Moreover, the very fact that one needs $n{-}2$ triples (to which we associate oriented triangles) to define a leading singularity and therefore the equivalence of two such objects requires $2(n{-}2)$ triangles gives the triangulation of the surface of a sphere with $n$ vertices and hence the name. We illustrate how a given on-shell diagram, that admits a sphere move, induces a triangulation in Figure~\ref{fig:intro1} that shows the emergence of a sphere.

\begin{figure}
\centering\includegraphics[scale=.88]{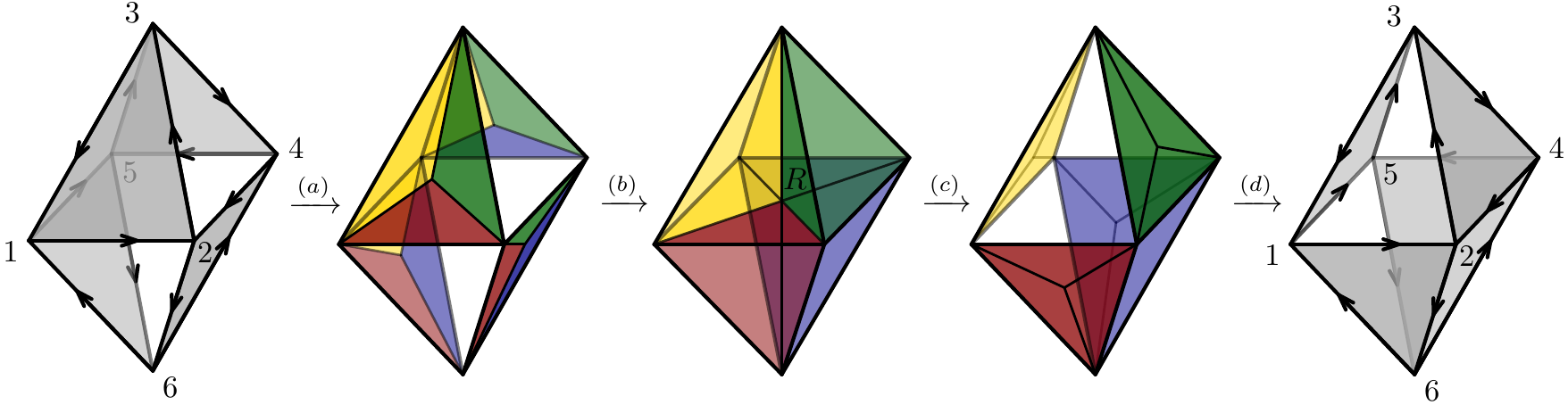}
\caption{\label{fig:intro2}Sphere move for the octahedron. (a) Start with $\mathcal{T}=\{(1,2,3),(3,4,5),(5,6,1),(6,4,2)\}$ and subdivide each triangle $\Delta_{abc}$ with a reference point (for clarity we remove orientations from the edges). (b) Move all four reference points into the common point $R$. (c) Open up the triangles into a different set of (shaded) triangles. (d) Remove the reference points to obtain the leading singularity associated to $\mathcal{T}=\{(2,3,4),(4,5,6),(6,1,2),(1,5,3)\}$.}
\end{figure}

The sphere move is performed by subdividing the triangle corresponding to $\Delta_{abc}$ into three new ones with a reference point $r$, such that:
\be
\Delta_{abc} = \Delta_{abr}+\Delta_{bcr}+\Delta_{car}.
\ee
Moving all reference points $r$ into a common one, say $R$, it is possible to open them up in a different order, such that the resulting triangles are the complement of the ones we started with. See Figure~\ref{fig:intro2} for an example.

In fact, it is most convenient to send the reference point $R$ to infinity, which gives variables associated to (oriented) edges:
\be
u_{ab} := \lim_{R \to \infty} \Delta_{abR},
\ee
with $u_{ab} = - u_{ba}$, so that $\Delta_{abc}=u_{ab}+u_{bc}+u_{ca}$. In this language, two sets of triangles are connected by a sphere move if their boundaries (i.e., $\Delta_{abc}$'s expanded as sums of $u_{ab}$'s) are the same. See Figure~\ref{fig:intro3} for an example of a square move as the simplest case of a sphere move.

\begin{figure}
\centering\includegraphics[scale=0.9]{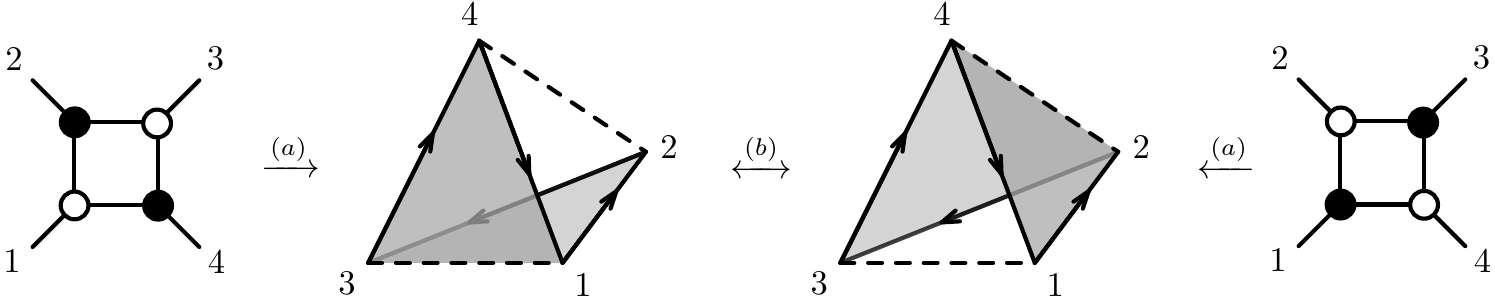}
\caption{\label{fig:intro3}Square move as the simplest example of a sphere move. (a) Translation between on-shell diagrams and triangulations of spheres. (b) On the LHS we have $\mathcal{T}_L=\{(1,2,3),(1,3,4)\}$, while on the RHS we have $\mathcal{T}_R=\{(1,2,4),(2,3,4)\}$. Both of their boundaries give the same 1-complex associated to the oriented edge set $\{(12),(23),(34),(41)\}$. Note that edges $(13)$ and $(24)$ cancel on both sides. Algebraically, we have $\Delta_{123}+\Delta_{134} = \Delta_{124}+\Delta_{234} = u_{12} + u_{23} + u_{34} + u_{41}$.}
\end{figure}

Many of the applications we find of the $\Delta$-algebra follow from the fact that one can use a triangulation independent formulation as the one introduced by Enciso \cite{Enciso:2014cta,Enciso:2016cif} and used in the mathematical construction of configuration spaces \cite{early2018configuration}. 

On-shell diagrams have been the subject of many studies and there are many properties which are understood with very elegant proofs \cite{Hodges:2009hk}, for reviews see \cite{Feng:2011np,Elvang:2013cua}. Here we present the $\Delta$-algebra formulation some of these properties, e.g., the $U(1)$ decoupling identity and the fundamental Bern--Carrasco--Johansson (BCJ) relation \cite{Bern:2008qj}. Moreover, the construction naturally leads to higher order identities such as those expected from the string theory formulation in terms of amplitudes involving $\text{Tr}F^4$ amplitudes \cite{BjerrumBohr:2009rd,BjerrumBohr:2010zs}. In our derivation, the objects that appear are double-trace-like instead. 

The fact that the determinant formula \eqref{rtnl}, as shown by Franco, Galloni, Penante, and Wen \cite{Franco:2015rma}, and Enciso's construction can both be generalized to higher $k$, motivated us to also generalize the $\Delta$-algebra to higher $k$. Here we only start the study of the object and explain some of the most basic properties.

This paper is organized as follows. In Section~\ref{sec:leading-singularities} we introduce the $\Delta$-algebra and explain its connection to leading singularities. In Section~\ref{sec:sphere-move} we use it to define sphere moves between classes of on-shell diagrams. In Section~\ref{sec:identities} we show identities satisfied by the $\Delta$-algebra, which allow us to prove known and new relations among leading singularities. Longer proofs needed in this section are relegated to the appendices. In Section~\ref{sec:higher-k} we discuss higher-$k$ generalizations of the $\Delta$-algebra. We end with a discussion of future directions in Section~\ref{sec:discussion}.

\section{\label{sec:leading-singularities}Leading Singularities: Introducing $\Delta_{abc}$}

Leading singularities in ${\cal N}{=}4$ super Yang--Mills are the most basic IR finite quantities in the theory \cite{ArkaniHamed:2009dn}. They are computed using on-shell diagrams and are classified by their R-charge in sectors \cite{ArkaniHamed:2012nw}. The simplest and most well-understood sector is the maximal helicity violating (MHV) one. All $n$-particle MHV leading singularities can be constructed from a set of $n{-}2$ cyclically-ordered triples of distinct labels ${\cal T} = \{ (a_1,b_1,c_1),\ldots ,(a_{n-2},b_{n-2},c_{n-2})\}$. The explicit form of a given leading singularity can be compactly encoded in the reduced determinant of a $(n{-}2)\times n$ matrix $M$ \cite{Arkani-Hamed:2014bca}. Rows are labeled by triples while columns by particles. A row corresponding to triple $(a_i,b_i,c_i)$ has only three non-zero components at columns $a_i$, $b_i$ and $c_i$ with values $\langle b_i,c_i\rangle, \langle c_i,a_i\rangle$ and $\langle a_i,b_i\rangle$ respectively. MHV leading singularities are known to depend only on holomorphic spinors $\lambda_a := \{\lambda_{a,1},\lambda_{a,2}\}$ through $SL(2,{\mathbb{C}})$ invariant combinations $\langle a,b\rangle$, which can be also thought of as Pl\"ucker coordinates of $G(2,n)$.

The matrix $M$ has two column null vectors with components $v^{(1)}_a=\lambda_{a,1}$ and $v^{(2)}_a=\lambda^{(1)}_{a,2}$. Using this it is possible to remove two columns, say the $a^{\text{th}}$ and $b^{\text{th}}$ columns, of $M$ to obtain a $(n{-}2)\times (n{-}2)$ submatrix $M^{(ab)}$ and show that
\begin{equation}
{\rm det}' M := \frac{{\rm det}\, M^{(ab)}}{\langle a,b\rangle}
\end{equation}
is independent of the choice of deleted columns.

The rational function associated to a list of triples $\cal T$ is then obtained as \eqref{rtnl}
\begin{equation}\label{ls}
F_{\cal T} = \frac{\left({\rm det}' M\right)^2}{\prod_{i=1}^{n-2}\langle a_i,b_i\rangle\langle b_i,c_i\rangle\langle c_i,a_i\rangle}.
\end{equation}

\subsection{Product of Triples}

In order to obtain a reformulation of leading singularities that leads to the $\Delta$-algebra, the first step is to use the well-known  formulation of determinants in terms of Grassmann variables.\footnote{Recall that any pair of Grassmann variables anticommutes, $\{\theta_a,\theta_b\} = 0$, $\{\bar\theta_a,\bar\theta_b\} = 0$, $\{\bar\theta_a,\theta_b\} = 0$, and the integration rule is $\int d\theta_a (\alpha + \beta \theta_b) = \beta \delta_{ab}$ for constants $\alpha,\beta \in \mathbb{C}$, see, e.g., \cite{deligne1999quantum} for a review.} Given any $m{\times}m$ matrix $R$, the determinant of $R$ can be expresses as an integral over two sets of Grassmann variables $\theta_a$ and $\bar\theta_a$ as follows
\be
{\rm det}\, R = \int \prod_{a=1}^m d\theta_a d\bar\theta_a {\rm exp}\left( \sum_{a,b=1}^m \theta_a R_{ab}\bar\theta_b \right).
\ee

It will be useful to carry out the integration over either set of Grassmann variables, say $\bar\theta$ to obtain
\be\label{pro}
{\rm det}\, R = \int \prod_{a=1}^m d\theta_a \prod_{c=1}^m \left(\sum_{b=1}^m \theta_b R_{bc}\right).
\ee

Let us apply \eqref{pro} to the determinant entering in the formula for rational function associated to a leading singularity \eqref{ls}. One can write
\be
{\rm det}' M =\frac{{\rm det}\, M^{(de)}}{\langle d,e\rangle} = \int \prod_{a=1}^n d\theta_a\, \prod_{i=1}^{n-2}\left(\theta_{a_i}\langle b_i,c_i\rangle+\theta_{b_i}\langle c_i,a_i\rangle+ \theta_{c_i}\langle a_i,b_i\rangle \right) \times \frac{\theta_d\theta_e}{\langle d,e\rangle}.
\ee

Note that the choice of columns that are removed is done by the choice of Grassmann variables in the last factor. Indeed, since $\int d\theta_a \, \theta_a = 1$, it is useful to think about $\theta = \delta(\theta)$ and therefore the last factor can be written as $\delta(\theta_d)\delta(\theta_e)$ which then imply that $\theta_e$ and $\theta_d$ must be set to zero in the rest of the integrand thus removing the corresponding columns.

Finally, since the leading singularity formula \eqref{ls} has two powers of the determinant, one can introduce a different set of Grassmann variables, say $\chi_a$, to write it. This shows that a leading singularity associated to a set of triples ${\cal T}$ is given by
\be
F_{\cal T} = \int \prod_{a=1}^n d\theta_a d\chi_a \, \prod_{i=1}^{n-2}\Delta_{a_ib_ic_i} \times \frac{\theta_d\theta_e}{\langle d,e\rangle}\frac{\chi_f\chi_g}{\langle f,g\rangle},
\label{LS-map}
\ee
where
\be
\Delta_{abc}:=\frac{\left(\theta_{a}\langle b,c\rangle+\theta_{b}\langle c,a\rangle+ \theta_{c}\langle a,b\rangle \right)\left(\chi_{a}\langle b,c\rangle+\chi_{b}\langle c,a\rangle+\chi_{c}\langle a,b\rangle \right)}{\langle a,b\rangle\langle b,c\rangle\langle c,a\rangle}.
\ee

The rest of the paper is devoted to studying properties of $\Delta_{abc}$'s as generators of a commutative algebra and its applications to physical quantities.

\subsection{From Homogeneous to Inhomogeneous Variables}

The variables $\lambda_a$ can be thought of as homogeneous variables of points on $\mathbb{CP}^1$. In physics applications, such as Cachazo--He--Yuan (CHY) \cite{Cachazo:2013hca} and Witten--Roiban--Spradlin--Volovich formulas \cite{Witten:2003nn,Roiban:2004yf} , it is more convenient to work with inhomogeneous variables. In order to go from one to the other it is enough to write
\be
\lambda_a = \left(
              \begin{array}{c}
                \lambda_{a,1} \\
                \lambda_{a,2} \\
              \end{array}
            \right) = t_a \left(
              \begin{array}{c}
                1 \\
                x_a \\
              \end{array}
            \right)
\ee
and $(\theta_a,\chi_a)\to t_a(\theta_a,\chi_a)$.

Under these operations one finds that $\Delta_{abc}$ has a much more compact form
\be\label{inho}
\Delta_{abc} = -\frac{\left(\theta_{a}x_{bc}+\theta_{b}x_{ca}+ \theta_{c}x_{ab}\right)\left(\chi_{a}x_{bc}+\chi_{b}x_{ca}+ \chi_{c}x_{ab}\right)}{x_{ab}x_{bc}x_{ca}},
\ee
where we have introduced the shorthand notation $x_{ab}:=x_a-x_b$.

Given that the expression for a given leading singularity is independent of the choice of columns deleted, i.e., the choice of $d,e,f$ and $g$ in the factor
\be
\frac{\theta_d\theta_e}{x_{de}}\frac{\chi_f\chi_g}{x_{fg}}
\ee
we choose to refer to a leading singularity as simply the product of $n{-}2$ $\Delta_{abc}$
\be
{\rm LS}_{\cal T}:= \prod_{\tau\in {\cal T}}\Delta_{\tau_1,\tau_2,\tau_3}.
\ee
It is straightforward to see that the map \eqref{LS-map} between $\text {LS}_{\cal T}$ and $F_{\cal T}$ is an isomorphism.
This is the formula presented in the introduction in \eqref{DLS}. Of course, the physically relevant object is the rational function $F_{\cal T}$ associated with it which is obtained after integrating out the Grassmann variables.

\subsection{Properties of $\Delta_{abc}$}

The building block of leading singularities has several crucial properties. First, $\Delta_{abc}$ is completely antisymmetric in its indices. Second, once again one can use that for any Grassmann variable $\theta = \delta(\theta)$ to conclude that
\be
\Delta_{abc}  = -\frac{\delta\left(\theta_{a}x_{bc}+\theta_{b}x_{ca}+ \theta_{c}x_{ab}\right)\delta\left(\chi_{a}x_{bc}+\chi_{b}x_{ca}+ \chi_{c}x_{ab}\right)}{x_{ab}x_{bc}x_{ca}}
\ee
and therefore $\Delta_{abc}^2 = 0$. Moreover, since the Grassmann degree is two, they commute $\Delta_{abc}\Delta_{efg}=\Delta_{efg}\Delta_{abc}$.

Finally, a small amount of algebra reveals that
\be\label{exu}
\Delta_{abc} = u_{ab} + u_{bc} + u_{ca}
\ee
with
\be\label{defU}
u_{ab}:= \frac{\theta_{ab}\chi_{ab}}{x_{ab}}.
\ee
Here $\theta_{ab}:=\theta_a-\theta_b$ and $\chi_{ab}:=\chi_a-\chi_b$. Note that $u_{ab}=-u_{ba}$ while, once again, the nature of Grassmann variables implies that $u_{ab}^2=0$ while $\Delta_{abc}^2 = 0$ shows that
\be
u_{ab}u_{bc} + u_{bc}u_{ca}+u_{ca}u_{ab} = 0.
\ee
These relations for $u_{ab}$ are exactly those defining the commutative algebra governing the cohomology ring of the configuration space, $\text{H}^\ast(\text{Conf}_n(SU(2))\slash SU(2),\mathbb{C})$, of $n$ distinct points in $SU(2)$ modulo the diagonal action was constructed as a subalgebra of the cohomology ring $\text{H}^\ast(\text{Conf}_n(\mathbb{R}^3),\mathbb{C})$ as found in \cite{early2018configuration}. In fact, our main motivation was to find a physical representation of this algebra.

It is also possible to derive the expression \eqref{exu} using the intuition developed by Enciso \cite{Enciso:2014cta,Enciso:2016cif} in his construction of a triangulation independent version of amplituhedron formulas. In order to see this, one starts by showing that $\Delta_{abc}$ satisfies properties of the boundary of an oriented triangle. Without assuming \eqref{exu} one can easily show from \eqref{inho} that
\be
\Delta_{abc} = \Delta_{abr}+\Delta_{bcr}+\Delta_{car}.
\label{eq:triangulationmain}
\ee
Note that since the LHS does not depend on $x_r$, it must be that neither does the RHS and therefore one can set $x_r$ to any value. Taking the limit $x_r\to \infty$ gives \eqref{exu} and we discover that 
\be
u_{ab} = \lim_{x_r\to \infty} \Delta_{abr}.
\ee

This point of view will be very useful in generalizations to higher values of $k$, as discussed in Section~\ref{sec:higher-k}. This concludes our presentation of the $\Delta$-algebra and we turn to physical applications. 

\section{\label{sec:sphere-move}Sphere Moves for Non-Planar On-Shell Diagrams}

The first application of the $\Delta$-algebra formulation of leading singularities has to do with their graphical representation as on-shell diagrams. Associated with a given leading singularity function there can be many different lists of triples, each representing an on-shell diagram. It is well-known that any two on-shell diagrams related by a square move give rise to the same physical object \cite{postnikov2006total,ArkaniHamed:2012nw}. In the planar case, this move, combined with expansion and contraction moves of like-colored vertices, is enough to define equivalence classes of physical objects. Here we find that the square move is the simplest example of an infinite class of moves that become available once non-planarity is allowed. 

We start by recalling that a leading singularity is represented by a set of triples
\be
{\cal T} = \{ (a_1,b_1,c_1), (a_2,b_2,c_2), \ldots ,(a_{n-2},b_{n-2},c_{n-2})\}
\ee
is computed as
\be
{\rm LS}_{\cal T} = \prod_{i=1}^{n-2}\Delta_{a_ib_ic_i}.
\ee
This object can be rewritten in a variety of ways. One possibility is by choosing complex numbers $\alpha_i$ such that $\alpha_1\alpha_2\cdots \alpha_{n-2}=1$ to write the identity
\be
\prod_{i=1}^{n-2}\Delta_{a_ib_ic_i} =\frac{1}{(n-2)!}\left(\sum_{i=1}^{n-2}\alpha_i \Delta_{a_ib_ic_i}\right)^{n-2},\label{power-identity}
\ee
where we used that $\Delta$'s commute and square to zero. Motivated by the expression of $\Delta_{abc}$ in terms of $u_{ab}$ given in \eqref{exu}, it is natural to choose $\alpha_i$'s such that as many $u_{ab}$'s can be canceled as possible.

Consider the simplest non-trivial example at $n{=}4$ for which one choice of triples is ${\cal T} = \{(1,2,3),(1,3,4)\}$ so that
\be
{\rm LS}_{\cal T} = \Delta_{123}\Delta_{134} = \frac{1}{2}\left(\alpha_1(u_{12}+u_{23}+u_{31})+ \alpha_2(u_{13}+u_{34}+u_{41})\right)^2
\ee
Since none of the $\alpha$'s can vanish, the only $u_{ab}$ that can be eliminated is $u_{13}$ whose coefficient is $\alpha_2-\alpha_1$ (recall that $u_{31}=-u_{13}$). Therefore asking this to cancel leads to $\alpha_1 = \alpha_2 = \pm 1$ and to a very illuminating formula
\be\label{four}
\Delta_{123}\Delta_{134} =\frac{1}{2}(u_{12}+u_{23}+u_{34}+u_{41})^2.
\ee
This form has a cyclic symmetry that the triples obscured. In particular, it means that
\be
\Delta_{123}\Delta_{134} = \Delta_{234}\Delta_{241},
\ee
which is known as the square move in the on-shell diagram representation of MHV leading singularities.

Given these properties it is clear that $u_{ab}$ should be associated with an oriented edge $E_{ab}$ of a graph connecting vertices $a$ and $b$, and hence $\Delta_{abc}$ should be associated to the boundary $\partial T_{abc} = E_{ab} + E_{bc} + E_{ca}$ of an oriented triangle (2-simplex) $T_{abc}$,
\begin{align}
u_{ab} &\;\leftrightarrow\; E_{ab},\nonumber\\
\Delta_{abc} &\;\leftrightarrow\; \partial T_{abc}.\label{triangle-identification}
\end{align}

Generalizing the LHS of the four-point formula \eqref{four} to any $n$ is trivial if we want to preserve the cyclic symmetry by replacing a square by an $n$-gon
\be\label{polyk}
\frac{1}{(n-2)!}(u_{12}+u_{23}+u_{34}+u_{45}+\ldots +u_{n-1,n}+u_{n1})^{n-2}.
\ee
It is clear that there are as many representation of this object in term of products of $n{-}2$ $\Delta$'s as the number of triangulations of an $n$-gon, the Catalan number $C_{n-2}$. 


The rational functions associated with these polygons are the most basic examples of MHV leading singularities and are known as Parke--Taylor factors \cite{Parke:1986gb}. Once all integrations over the Grassmann variables are done the final result is \cite{Arkani-Hamed:2014bca}:
\be
{\rm PT}(1,2,\ldots, n):= \frac{1}{x_{12}x_{23}\cdots x_{n1}}.\label{Parke-Taylor}
\ee
Here we used the opportunity of introducing the notation for a Parke--Taylor factor.

Motivated by the above discussion we make the following definition.
\begin{definition}\label{move-definition}
Two distinct sets of $n{-}2$ triples ${\cal T}_1$ and ${\cal T}_2$ are said to be \textit{equivalent} if they give rise to the same leading singularity, up to a sign,
\be
\text{LS}_{{\cal T}_1} = \pm\,\text{LS}_{{\cal T}_2}.
\ee
By equation \eqref{power-identity}, a sufficient condition for this to happen is that there exist non-zero coefficients $\alpha_\tau,\beta_\tau $ such that
\be
\sum_{\tau \in {\cal T}_1} \alpha_\tau \Delta_{\tau_1,\tau_2,\tau_3} = \sum_{\tau \in {\cal T}_2} \beta_\tau \Delta_{\tau_1,\tau_2,\tau_3}.
\ee
\end{definition}

In the sequel we explore a particular infinite family of equivalences between non-planar on-shell diagrams which have a nice geometric interpretation.

\subsection{The Sphere Move}

The first example of a new move was found by one of the authors in \cite{Early:2018mac}, where the $6$-point non-planar leading singularity associated with triples $\mathcal{T} = \{(1,2,3),(3,4,5),(5,6,1),(6,4,2)\}$ was found to be invariant under a cyclic shift of labels, $i \to i+1$. This property does not follow from square moves as none are possible for this set of triples. Using the results of the previous discussions it is easy to see the cyclic property by writing
\be
\text{LS}_{\cal T}= \Delta_{123}\Delta_{345}\Delta_{561}\Delta_{642} = \frac{1}{4!}\left(\Delta_{123}+\Delta_{345}+\Delta_{561}+\Delta_{642} \right)^4
\ee
and using that 
\begin{align}
\Delta_{123}+\Delta_{345}+\Delta_{561}+\Delta_{642} &= u_{12}{+}u_{23}{+}u_{34}{+}u_{45}{+}u_{56}{+}u_{61}{+}u_{15}{+}u_{26}{+}u_{31}{+}u_{42}{+}u_{53}{+}u_{64}\nonumber\\
&= \Delta_{234} + \Delta_{456} + \Delta_{612} + \Delta_{153}.\label{octahedral-move}
\end{align}
In other words, the graph given by the $u_{ab}$'s is the oriented circulant graph $C_6(1,2)$ which is clearly cyclic invariant.

Let us explain how the square move and the one just discussed are the first two examples of a family of moves.

Consider a triangulation of a sphere. The simplest one gives rise to a tetrahedron. Each triangle, $T_{abc}$, has an orientation induced by that of the sphere. Let the vertices of the tetrahedron be labelled $\{1,2,3,4\}$. The oriented triangles are $\{T_{123},T_{134},T_{142},T_{243}\}$. Clearly, the boundary of the surface is empty. Using the boundary map $\partial T_{abc} = E_{ab} + E_{bc} + E_{ca}$ together with the identification \eqref{triangle-identification} one finds that
\be
(\Delta_{123}+\Delta_{134}) - (\Delta_{124}+\Delta_{234}) = 0.\label{square-move}
\ee
This is nothing but the square move, see Figure~\ref{fig:sphere-move}. It is clear that \eqref{octahedral-move} follows from a similar reasoning for a triangulation of a sphere containing the triangles $\{ T_{123}, T_{345}, T_{561}, T_{642}, T_{243},\allowbreak T_{465}, T_{621}, T_{135}\}$.

\begin{figure}[h!]
	\centering
	\includegraphics[scale=0.9]{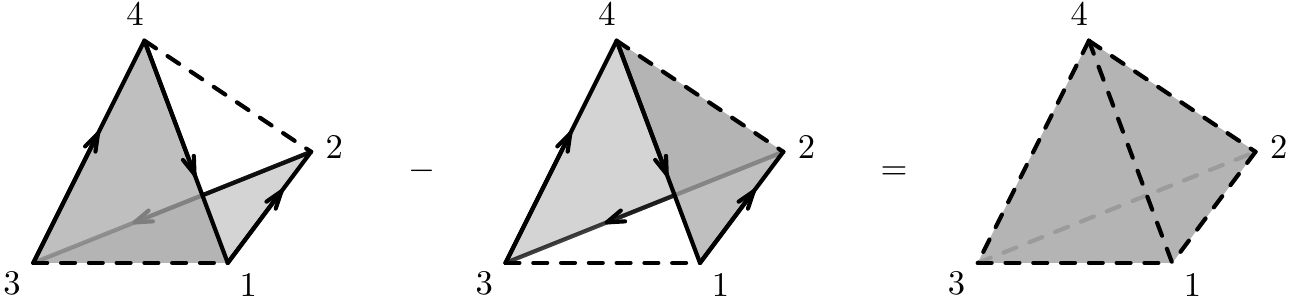}
	\caption{\label{fig:sphere-move}Square move: the simplest example of a sphere move.}
\end{figure}

A general triangulation of the sphere has $F$ triangular faces, $E =3F/2$ edges and $V=n$ vertices satisfying $V-E+F=2$. This gives exactly $F=2(n{-}2)$ triangles. One can then split the set of triangles into two disjoints sets of $n{-}2$ triangles. Each set gives rise to a list of $n{-}2$ triples and therefore to a leading singularity or to something that vanishes. In either case, one finds that both must give rise to the same object, up to a sign, as a consequence of the boundary of the surface of the sphere being empty. This is why we call these families of identities \emph{sphere} moves. Sphere moves can also be performed locally on a subset of triples from the two on-shell diagrams. This motivates the following definition.

\begin{definition}
Two sets of triples ${\cal T}_1$ and ${\cal T}_2$ are said to be related by a \emph{sphere move} of type ${\cal S}$, for a triangulation ${\cal S}$ of a $2$-sphere, if
\be\label{definition-sphere-move}
\sum_{\tau \in {\cal T}_1} T_{\tau_1, \tau_2, \tau_3} - \sum_{\tau \in {\cal T}_2} T_{\tau_1, \tau_2, \tau_3} = \sum_{\tau \in \mathcal{S}} T_{\tau_1, \tau_2, \tau_3},
\ee
where $T_{abc}$ is a triangle corresponding to the triple $(a,b,c)$. In particular, by the identification \eqref{triangle-identification} it implies that
\be
\sum_{\tau \in {\cal T}_1} \Delta_{\tau_1, \tau_2, \tau_3} = \sum_{\tau \in {\cal T}_2} \Delta_{\tau_1, \tau_2, \tau_3},
\ee
since the boundary of a sphere is empty. Note that here the number of triangles in ${\cal S}$ can be smaller than that of ${\cal T}_1 \cup {\cal T}_2$, i.e., sphere moves can be performed locally.
\end{definition}

We will often say that two sets of triples are related by a sphere move if there exists any ${\cal S}$ satisfying the above criteria.

At four points the only sphere move is the square move, see Figure~\ref{fig:sphere-move}. At five points one has the triangular bipyramid with triangles $\{(1,3,5),(4,3,1),(1,2,4),(3,4,2),(3,2,5),(1,5,2)\}$. One possible identity is obtained by splitting $\{(1,3,5),(4,3,1),(3,4,2)\}$ and $\{(1,2,4),(3,2,5),\allowbreak (1,5,2)\}$. Both give rise to the Parke--Taylor factor ${\rm PT}(1,5,3,2,4)$. Of course, this can be obtained by a sequence of two square moves and therefore it is not a new identity. Examples of sphere moves beyond the square one start at six points, which is what we will discuss next.

\subsection{Further Examples of Sphere Moves}

We start by giving an infinite family of examples that generalize the sphere move on the octahedral leading singularity in a natural way into bipyramids with more faces. Let us consider the following set of triples for $n$ even:
\begin{align}\label{bipyramid-move-1}
{\cal T}_1 = \{ &(1,2,n{-}1), (3,4,n{-}1), \ldots, (n{-}3,n{-}2,n{-}1),\nonumber\\
 &(3,2,n), (5,4,n), \ldots, (1,n{-}2,n)\}.
\end{align}
It is related by a sphere move to the following set of triples obtained by exchanging $n{-}1 \leftrightarrow n$:
\begin{align}\label{bipyramid-move-2}
{\cal T}_2 = \{ &(3,2,n{-}1), (5,4,n{-}1), \ldots, (1,n{-}2,n{-}1),\nonumber\\
&(1,2,n), (3,4,n), \ldots, (n{-}3,n{-}2,n)\}.
\end{align}
We illustrate this sphere move in Figure~\ref{fig:bipyramid-move}.

\begin{figure}[h!]
	\centering
	\includegraphics[scale=0.9]{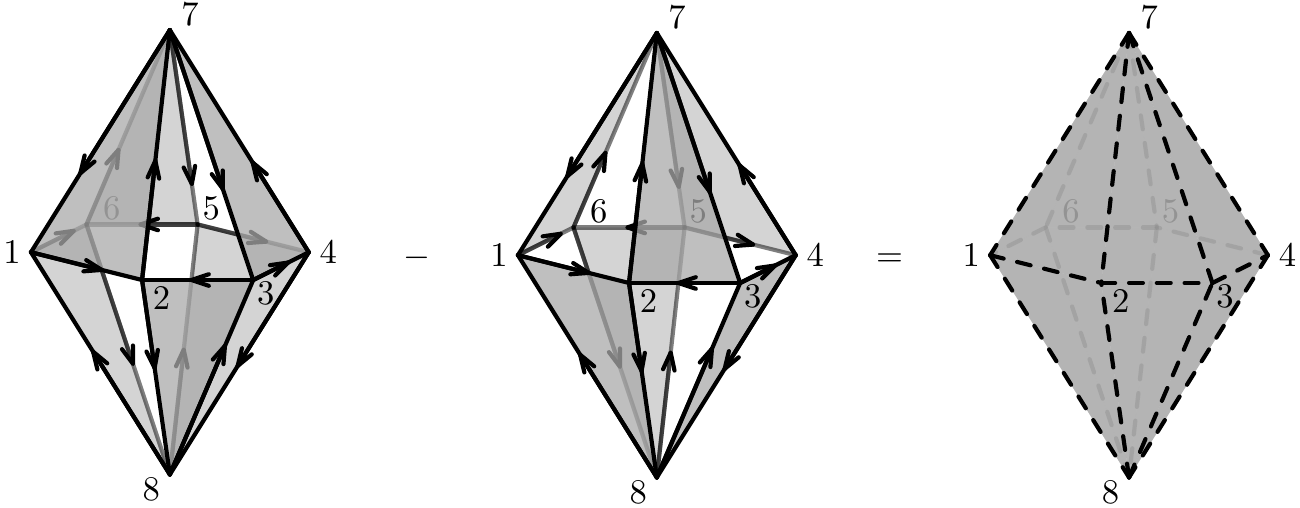}
	\caption{\label{fig:bipyramid-move}
		Sphere move between \eqref{bipyramid-move-1} and \eqref{bipyramid-move-2} for $n{=}8$. The corresponding sets of triples are ${\cal T}_1 = \{ (1,2,7), (3,4,7), (5,6,7), (3,2,8), (5,4,8), (1,6,8)\}$ and ${\cal T}_2 = \{ (3,2,7), (5,4,7), (1,6,7),\allowbreak (1,2,8),\allowbreak (3,4,8),\allowbreak (5,6,8)\}$.
	}
\end{figure}

The next family we describe is given by the set of triples, again for even $n=2m$:
\be\label{drum-move-1}
{\cal T}_1 = \{ (1,2,\ldots,m), (m{+}1,m{+}2,2), (m{+}2,m{+}3,3), \ldots, (2m,m{+}1,1)\}.
\ee
Here we abused the notation by using $(1,2,\ldots,m)$ in the first slot to denote an arbitrary triangulation of an $m$-gon with $m{-}2$ triples. The total number of triples in \eqref{drum-move-1} is therefore $n{-}2$, as expected for $n$ labels. It is related by a sphere move to the following set:
\be\label{drum-move-2}
{\cal T}_2 = \{(1,2,m{+}1), (2,3,m{+}1), \ldots, (m,1,2m), (m{+}1, m{+}2, \ldots, 2m) \},
\ee
which can be obtained from the original one by relabelling $i \to 2m{-}i{+}1$ and flipping orientations of all triples.
We give an example of this sphere move in Figure~\ref{fig:drum-move}.

\begin{figure}[h!]
	\centering
	\includegraphics[scale=0.9]{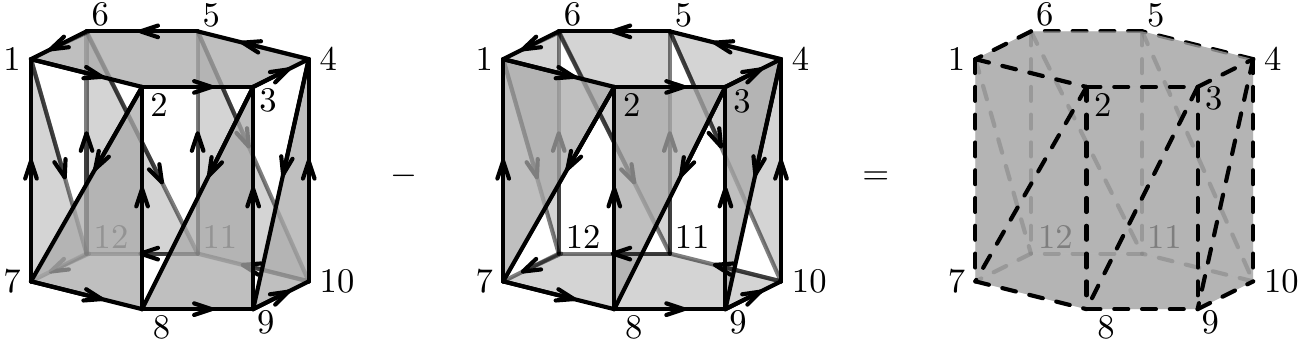}
	\caption{\label{fig:drum-move}
		Sphere move between \eqref{drum-move-1} and \eqref{drum-move-2} for $n{=}12$. The corresponding sets of triples are ${\cal T}_1 = \{ (1,2,3,4,5,6), (7,8,2), (8,9,3), (9,10,4), (10,11,5), (11,12,6), (12,7,1)\}$ and ${\cal T}_2 = \{ (1,2,7),\allowbreak (2,3,8),\allowbreak (3,4,9),\allowbreak (4,5,10), (5,6,11), (6,1,12), (7,8,9,10,11,12) \}$.
	}
\end{figure}

Clearly, one can generate an infinite number of new sphere moves by starting with an arbitrary triangulation (not necessarily invariant under any relabeling) of a sphere with $2(n{-}2)$ triangles and dividing it into two sets of $n{-}2$, which yields an equality between two leading singularities. We leave such classification questions to future work.

Let us close with an example of a sphere move that is performed locally and does not involve all the $n{-}2$ triangles. Let us start by considering the sets of triples
\begin{align}
{\cal T}_1  = \{(1, 2, 3), (3, 4, 5), (5, 6, 1), (6, 4, 2), (7, 8, 9), (9, 10, 6), (6, 5, 7), (8, 5, 10)\},\nonumber\\
{\cal T}_2 = \{(2,3,4),(4,5,6),(6,1,2),(1,5,3),(7, 8, 9), (9, 10, 6), (6, 5, 7), (8, 5, 10)\}.
\label{composite-sphere-move}
\end{align}
Applying the definition \eqref{definition-sphere-move} between ${\cal T}_1$ and ${\cal T}_2$ we find:
\begin{align}
\sum_{\tau \in {\cal T}_1} T_{\tau_1, \tau_2, \tau_3} - \sum_{\tau \in {\cal T}_2} T_{\tau_1, \tau_2, \tau_3} &= T_{123} + T_{345} + T_{561} + T_{642} - T_{234} - T_{456} - T_{612} - T_{153} \nonumber\\
&= \sum_{\tau \in \mathcal{S}} T_{\tau_1, \tau_2, \tau_3}.\label{two-octahedra-3}
\end{align}
Notice that in the first equality several terms have cancelled, such that only those involving the labels $\{1,2,3,4,5,6\}$ appear. The result is an octahedral triangulation $\mathcal{S}$ of a sphere with $6$ vertices, $8$ edges, and $8$ triangles. Hence ${\cal T}_1$ and ${\cal T}_2$ are connected by a sphere move, in analogy with the one given in \eqref{octahedral-move}.

Consider now another set of triples, given by
\be
{\cal T}_3 = \{(2,3,4),(4,5,6),(6,1,2),(1,5,3), (8,9,10),(10,6,5),(5,7,8),(9,7,6)\}.
\ee
The difference between two sets of triangles corresponding to ${\cal T}_2$ and ${\cal T}_3$ is
\begin{align}
\sum_{\tau \in {\cal T}_2} T_{\tau_1, \tau_2, \tau_3} - \!\! \sum_{\tau \in {\cal T}_3} T_{\tau_1, \tau_2, \tau_3} &= T_{7,8,9} + T_{9,10,6} + T_{6,5,7} + T_{8,5,10} - T_{8,9,10} - T_{10,6,5} - T_{5,7,8} - T_{9,7,6} \nonumber\\
&= \sum_{\tau \in \mathcal{S}} T_{\tau_1, \tau_2, \tau_3},\label{two-octahedra-4}
\end{align}
where once again the labels surviving on the RHS are only those from the set $\{5,6,7,8,9,10\}$ and they can also be understood as a triangulation of a sphere. Hence ${\cal T}_2$ and ${\cal T}_3$ are connected by a sphere move. As a result of this ${\cal T}_1$ and ${\cal T}_3$ are in fact equivalent. Note that each of them can be obtained by gluing two octahedra by an edge. Hence the corresponding triangles cannot be embedded on a surface of a sphere. 

Alternatively, one can arrive at the set of triples ${\cal T}_3$ from ${\cal T}_1$ by a composition of two different sphere moves. Using the intermediate set
\be
\widetilde{\cal T}_2 = \{(1, 2, 3), (3, 4, 5), (1,7,6), (6, 4, 2), (7, 8, 9), (9, 10, 6), (1,5,7), (8, 5, 10)\}
\ee
we find:
\begin{align}\label{two-octahedra-1}
\sum_{\tau \in {\cal T}_1} T_{\tau_1, \tau_2, \tau_3} - \sum_{\tau \in \widetilde{\cal T}_2} T_{\tau_1, \tau_2, \tau_3} &= T_{561} + T_{657} - T_{176} - T_{157} = \sum_{\tau \in \mathcal{S}} T_{\tau_1, \tau_2, \tau_3},
\end{align}
where we recognize that ${\cal S}$ is a tetrahedral triangulation of a sphere, i.e., ${\cal T}_1$ and $\widetilde{\cal T}_2$ are connected by a square move in the labels $\{1,5,7,6\}$, similar to the one given in \eqref{square-move}. Let us consider the difference between $\widetilde{\cal T}_2$ and ${\cal T}_3$:
\begin{align}
\sum_{\tau \in \widetilde{\cal T}_2} T_{\tau_1, \tau_2, \tau_3} - \sum_{\tau \in {\cal T}_3} T_{\tau_1, \tau_2, \tau_3} &= T_{1,2,3} + T_{3,4,5} + T_{1,7,6} + T_{6,4,2} + T_{7,8,9} + T_{9,10,6} + T_{1,5,7} + T_{8,5,10} \nonumber\\
&\quad - T_{2,3,4}\! - T_{4,5,6}\! - T_{6,1,2}\! - T_{1,5,3}\! - T_{8,9,10}\! - T_{10,6,5}\! - T_{5,7,8}\! - T_{9,7,6}\nonumber\\
&= \sum_{\tau \in \mathcal{S}} T_{\tau_1, \tau_2, \tau_3}.\label{two-octahedra-2}
\end{align}
Here the RHS involves all the available labels and ${\cal S}$ is a triangulation of a sphere with $10$ vertices, $24$ edges, and $16$ triangles. Hence $\widetilde{\cal T}_2$ and ${\cal T}_3$ are connected by a sphere move, which shows equivalence between ${\cal T}_1$ and ${\cal T}_3$ by a composition of two sphere moves. We illustrate this procedure in Figure~\ref{fig:composite-sphere-move}.

\begin{figure}[h!]
	\centering
	\includegraphics[scale=0.88]{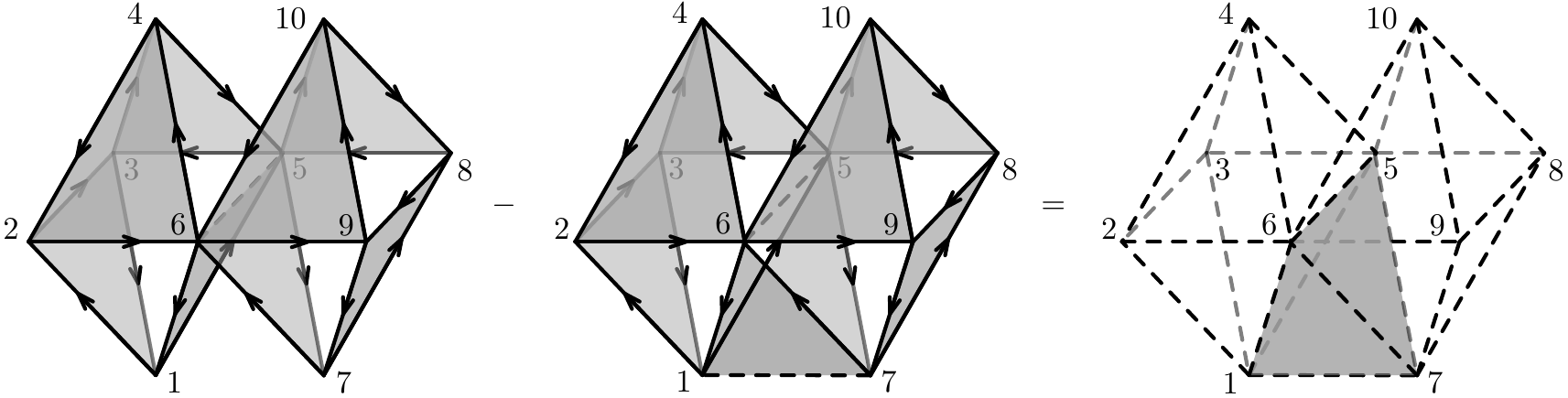}
	\includegraphics[scale=0.88]{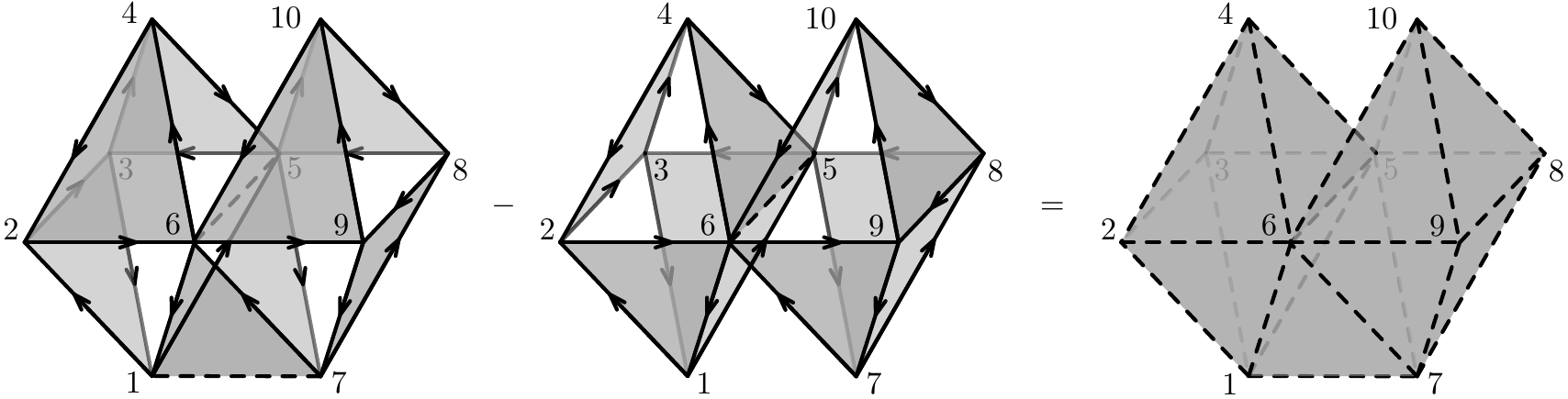}
	\caption{\label{fig:composite-sphere-move}A sequence of two sphere moves between the set of triangles from \eqref{composite-sphere-move}. Top: sphere (square) move on the labels $\{1,5,7,6\}$ from \eqref{two-octahedra-1}. Bottom: sphere move on all the labels from \eqref{two-octahedra-2}.}
\end{figure}

Indeed, it can be confirmed directly from the determinant formula \eqref{LS-map} that the rational functions $F_{{\cal T}_i}$ for $i=1,2,3$ are all equal to
\be
-\frac{\left(x_{15} x_{26} x_{34}-x_{16} x_{24} x_{35}\right)^2 \left(x_{58}
   x_{6,10} x_{79}-x_{5,10} x_{69} x_{78}\right)^2}{x_{12} x_{13} x_{15} x_{16}
   x_{23} x_{24} x_{26} x_{3,4} x_{3,5} x_{45} x_{46} x_{57} x_{58} x_{5,10}
   x_{67} x_{69} x_{6,10} x_{78} x_{79} x_{89} x_{8,10} x_{9,10}}.\label{two-octahedra-rtnl}
\ee
The numerator factors, which is a consequence of the fact that the sets of triples can be obtained by gluing two leading singularities. In fact, in Appendix~\ref{app:gluing} we show that this is a generic property of any leading singularity that contains a smaller one inside of it.

\section{\label{sec:identities}Parent Identities for Triples}

In this section we present new identities among sets of triples, when realized in terms of $\Delta$ variables. These identities are the lower order analogues of well-known relations between leading singularities. In fact, we show that these correspond to parent identities in the sense that the $U(1)$ decoupling, Bern--Carrasco--Johansson (BCJ) \cite{Bern:2008qj} and Kleiss--Kuijf (KK) relations \cite{Kleiss:1988ne} follow almost directly from the low order relations. The scope of these relations between $\Delta$ objects is however far more general. We end section presenting new extensions of the $U(1)$ decoupling identity involving different types of leading singularities.
We will concentrate on the study and applications of expressions of the form
\be
\, \left( \sum_{i=1}^{n-2}\alpha_i \Delta_{a_ib_ic_i} \right)^j
\label{eq:powerj}
\ee
We will interpret the $x_i$ variables appearing in the $\Delta$'s as the location of punctures on a $\mathbb{CP}^1$. The power $j= n{-}2$ of this expression can be taken as an
integrand for the CHY construction \cite{Cachazo:2013hca}. Relations among integrands are then mapped to relations among amplitudes in theories which admit a CHY representation.

\subsection{Decoupling of an Internal Label}

We have discussed how different on-shell diagrams can give rise to the same leading singularity function. It is also known that leading singularity functions are not linearly independent. The simplest relation is the so-called $U(1)$ decoupling identity for $n=4$. This is the following three-term identity
\be
{\rm PT}(1234) + {\rm PT}(2314) + {\rm PT}(3124) = 0.
\ee
It is natural to ask for a $\Delta$-algebra interpretation of these kind of relations. Using the following representation of the left hand side in terms of $\Delta$'s reveals the general structure
\be\label{Uden}
\Delta_{124}\Delta_{234} + \Delta_{234}\Delta_{314} + \Delta_{314}\Delta_{124} = \frac{1}{2}( \Delta_{124}+\Delta_{234}+\Delta_{314} )^2.
\ee
Now it is obvious why this has to vanish. Using the triangulation property \eqref{eq:triangulationmain}, i.e.,
\be
\Delta_{124}+\Delta_{234}+\Delta_{314} = \Delta_{123}\,,
\ee
we conclude that its square vanishes.

In addition to making the identity a consequence of a simple cancellation, this discussion also reveals an important structure. Given the one-element list of triples ${\cal T} = \{(123)\}$ its leading singularity is simply $\Delta_{123}$. In \eqref{Uden} we found that $\Delta_{123}^2/2$ is also useful for deriving identities. This suggests that one can talk about leading singularities and identities in a unified way by introducing the following exponential formula
\be
{\rm LS}_{\cal T} = {\rm exp}(\kappa \Delta_{123})\Big|_{\kappa}.
\ee
Indeed, following \cite{Early:2018mac}, we propose that associated with any list of triples one should construct an exponential representation so that
\be
{\rm LS}_{\cal T} = \left.{\rm exp}\left(\kappa \sum_{i=1}^{n-2}\alpha_i \Delta_{a_ib_ic_i}\right)\right|_{\kappa^{n-2}}.
\ee
While it is clear that the coefficient of $\kappa^{n-2}$ computes the leading singularity function, the previous example shows that other orders of the expansion contain valuable information as well. We investigate these orders in the next section.

\subsubsection{$U(1)$ Decoupling}

Let us start with the $U(1)$ decoupling identity. The name of the identity refers to the fact that it is a property of a $U(N)$ Yang--Mills theory where one particle is taken to be in $U(1)\subset U(N)$. The Lagrangian of a $U(N)$ theory splits into two parts, one for the $U(1)$ ``photon'' and one for the $SU(N)$ ``gluons''. There is no cross term and therefore no coupling among the photon and the gluons, i.e., the photon is decoupled. This means that the scattering amplitude of $n$ gluons and a photon must vanish, see, e.g., \cite{Dixon:1996wi} for a review.

Let us start with the usual Parke--Taylor interaction of $n$ gluons described at the linear order by the $n$-gon
\be
d_{1,2,\ldots ,n} :=u_{12}+u_{23}+\cdots +u_{n-1,n}+u_{n1}.
\label{eq:linpt}
\ee
Clearly, the power $d_{1,2,\ldots, n }^{n-2}$, integrated over the Grassmann variables, computes the Parke--Taylor factor $\text{PT}(1,2,\ldots,n)$ as in \eqref{Parke-Taylor}. This combination of signed edges, $u_{ab}$, can be written in terms of the triangle boundaries $\Delta_{abc}$ in many ways, in particular, one can choose any triangulation of the $n$-gon. 

Now we want to include the photon as the particle $n{+}1$ in a way that manifestly decouples. This is done by realizing that any point in the interior of the $n$-gon can be used to define a triangulation. Since the point is in the interior nothing depends on it and it is ``decoupled''. Denoting the interior point as $n+1$ one finds
\be
u_{12}+u_{23}+\ldots +u_{n-1,n}+u_{n1} = \Delta_{1,2,n+1}+\Delta_{2,3,n+1}+\cdots +\Delta_{n-1,n,n+1}+\Delta_{n,1,n+1}.
\label{eq:lineardecoupling}
\ee
The $(n{-}1)^{\text{th}}$ power vanishes as the left hand side can be written in terms of $n{-}2$ $\Delta$'s, e.g.,
\be
u_{12}+u_{23}+\ldots +u_{n-1,n}+u_{n1} =
\Delta_{1,2,n}+\Delta_{2,3,n}+\cdots + \Delta_{n-2,n-1,n}\,.
\label{eq:lineardecoupling2}
\ee
and therefore saturates at the power $n{-}2$. Expanding the right hand side of \eqref{eq:lineardecoupling} gives the identity we are looking for. Note that the right hand side has $n$ $\Delta$'s and therefore the $(n{-}1)^{\text{th}}$ power is a sum over $n$ terms in which exactly one of the $\Delta$'s is missing. This can be written as
\begin{align}
d^{n-1}_{1,2,\ldots ,n,n+1} &+ d^{n-1}_{1,2,\ldots ,n+1,n} +\cdots +d^{n-1}_{1,2,n+1,3\ldots ,n}+d^{n-1}_{1,n+1,2,\ldots ,n} = 0.
\label{eq:u1dpower}
\end{align}
After Grassmann integration each term gives rise to a $(n{+}1)$-particle Parke--Taylor factor leading to
\begin{align}
{\rm PT}(1,2,\ldots ,n,n{+}1) &+ {\rm PT}(1,2,\ldots ,n{+}1,n)+\nonumber\\
&+\cdots +{\rm PT}(1,2,n{+}1,3\ldots ,n)+{\rm PT}(1,n{+}1,2,\ldots ,n) = 0.
\end{align}

Given the linear realization of the Parke--Taylor factors defined in \eqref{eq:linpt} it is natural to ask if there is an analogue of \eqref{eq:u1dpower} for lower orders. The answer is yes. Indeed, let us now show how to generalize \eqref{eq:u1dpower} to all powers of $d$'s. Denoting $d_i := d_{1,\ldots i ,n+1,i+1,\ldots n}$ and using \eqref{eq:lineardecoupling} it is easy to see that we can write 
\begin{equation}
d_i =  d_{1,2,\ldots ,n} + \Delta_{i,n+1,i+1}.
\end{equation}
Now consider the sum 
\begin{align}
e^{\kappa\, d_1}+ e^{\kappa\, d_2}+ \ldots + e^{\kappa\, d_n} &=  e^{\kappa\, d_{1,2,\ldots, n}}\left( e^{\kappa\, \Delta_{1,n+1,2}} + e^{\kappa\, \Delta_{2,n+1,3}}+ \ldots + e^{\kappa\, \Delta_{n,n+1,1}} \right) \nonumber \\
&=  e^{\kappa \, d_{1,2,\ldots, n}}\left( n - \kappa\, d_{1,2,\ldots,n} \right).
\label{eq:expd}
\end{align}
In the second line we used $e^{\kappa \Delta}= 1 + \kappa\, \Delta$. Now, using the fact that $d_{1,2,\ldots,n}^{n-1}=0$ we conclude that the right hand side is at most of order $\kappa^{n-2}$. Extracting the order $\kappa^{n-1}$ in the left hand side then recovers the standard $U(1)$ decoupling \eqref{eq:u1dpower}. However, we now see that this is just the first of a family of identities that follow from extracting the lower orders in $\kappa$. At order $\kappa^j$ we have
\begin{align}
d^{j}_{1,2,\ldots ,n,n+1} &+ d^{j}_{1,2,\ldots ,n+1,n} +\cdots +d^{j}_{1,2,n+1,3\ldots ,n}+d^{j}_{1,n+1,2,\ldots ,n} = (n-j)\, d_{1,2,\ldots ,n}^j.
\end{align}

Because these relations follow from an straightforward application of \eqref{eq:lineardecoupling} we refer to the latter as the ``parent identity'' for $U(1)$ decoupling. Next we study similar identities for BCJ and KK relations.

\subsubsection{Fundamental BCJ Identity}

Here we present the derivation of the fundamental Bern--Carrasco--Johansson relations \cite{Bern:2008qj}. In theories which admit a BCJ representation, these relations reduce the number of independent partial amplitudes, corresponding to Parke--Taylor factors, to $(n{-}3)!$. Here these identities arise naturally from a generalization of the decoupling procedure at the linear level in $\Delta$.  

We are now in position to derive BCJ relations from a generalization of the decoupling identity \eqref{eq:lineardecoupling}. We will need to relate the coefficients $\alpha_i$ in the general form \eqref{eq:powerj} to kinematic invariants. The starting point is to note that the decoupling condition of particle $n{+}1$ does \textit{not} require the coefficients of
\be
\Delta_{1,2,n+1}+\Delta_{2,3,n+1}+\cdots +\Delta_{n-1,n,n+1}+\Delta_{n,1,n+1}
\ee
to be necessarily one. In order to find the most general setup for this to happen, let us consider the linear combination
\be\label{genkin}
  A_{12}\Delta_{1,2,n+1}+A_{23}\Delta_{2,3,n+1}+\cdots +A_{n-1,n}\Delta_{n-1,n,n+1}+A_{n1}\Delta_{n,1,n+1} = I_{n+1} + \Gamma_{n+1}\,,
\ee
where have decomposed the LHS into two parts. In the term  $\Gamma_{n+1}$ we have collected the dependence on particle $n{+}1$, i.e.,
\be\label{comB}
\Gamma_{n+1}:=(A_{12}-A_{n1})u_{n+1,1} + (A_{23}-A_{12})u_{n+1,2}+\cdots +(A_{n1}-A_{n-1,n})u_{n+1,n},
\ee
whereas the rest reads
\be\label{comd}
I_{n+1} := A_{12}\, u_{12}+A_{23}\,u_{23}+\cdots +A_{n-1,n}\, u_{n-1,n}+A_{n1}\,u_{n1}.
\ee
The linear expression \eqref{genkin} will be the parent identity generating BCJ relations (and other new relations among CHY integrands), provided the coefficients $A_{i,i+1}$ satisfy a decoupling condition. Here we will be mainly interested in the combination \eqref{genkin} raised to its $(n{-}1)^{\rm th}$ power. 

When computing the leading singularity the only non-trivially zero terms are those for which the Grassmann variables saturate after the gauge fixing term is included. In other words, in the expression
\be
(A_{12}\Delta_{1,2,n+1}+A_{23}\Delta_{2,3,n+1}+\cdots +A_{n-1,n}\Delta_{n-1,n,n+1}+A_{n1}\Delta_{n,1,n+1})^{n-1}\frac{\theta_d\theta_e}{x_{de}}\frac{\chi_f\chi_g}{x_{fg}}
\ee
the only potentially non-zero terms are those proportional to $\prod_{i=1}^{n+1}\theta_i\chi_i$. If one chooses $d,e,f,g$ to be different from $n{+}1$, then the only source of $\theta_{n+1}\chi_{n+1}$ is the combination $\Gamma_{n+1}$ defined in \eqref{comB}. In the following we will thus consider the expansion:
\be
(I_{n+1} + \Gamma_{n+1})^{n-1} = (I_{n+1})^{n-1}+(n-1)(I_{n-1})^{n-2}\Gamma_{n+1} + {\cal O}(\Gamma_{n+1}^2).
\label{eq:expigamma}
\ee
Clearly, the first term in the left hand side vanishes after Grassmann integration as it does not have any factors of $\theta_{n+1}\chi_{n+1}$.

We anticipate that each order in $\Gamma_{n+1}$ will give us a relation among Parke--Taylor factors and other CHY integrands. Let us start by considering the linear part, i.e., the second term in \eqref{eq:expigamma}. The combination $\Gamma_{n+1}$ must provide all the dependence in $\theta_{n+1}\chi_{n+1}$ and therefore we can drop all other Grassmann variables inside it. Using that $u_{ab}=\theta_{ab}\chi_{ab}/x_{ab}$ one finds
that the relevant piece of $\Gamma_{n+1}$ is
\be
 \left(\frac{(A_{12}-A_{n1})}{x_{n+1,1}} + \frac{(A_{23}-A_{12})}{x_{n+1,2}}+\cdots +\frac{(A_{n1}-A_{n-1n})}{x_{n+1,n}}\right)\theta_{n+1}\chi_{n+1}.
\ee
This means that in order to have an identity among Parke--Taylor functions it is necessary for this to vanish. More explicitly, if
\be\label{seq}
\frac{(A_{12}-A_{n1})}{x_{n+1,1}} + \frac{(A_{23}-A_{12})}{x_{n+1,2}}+\cdots +\frac{(A_{n1}-A_{n-1,n})}{x_{n+1,n}} = 0,
\ee
then
\begin{align}
\frac{1}{A_{n1}}{\rm PT}(1,2,\ldots ,n,n{+}1) &+ \frac{1}{A_{n-1,n}}{\rm PT}(1,2,\ldots ,n{+}1,n)+\nonumber\\
&+\cdots +\frac{1}{A_{1,2}}{\rm PT}(1,n{+}1,2,\ldots ,n) = {\cal O}(\Gamma_{n+1}^2).\label{fiden}
\end{align}

Clearly, $A_{ab} = 1$ is a solution to \eqref{seq} and gives rise to the standard $U(1)$ decoupling identity since $\Gamma_{n+1}$ vanishes identically.
This system can be written by performing the change of variables $A_{a,i+1}-A_{a-1,i} = s_{n+1,i}$, in which case it reads
\be
\frac{s_{n+1,1}}{x_{n+1,1}} + \frac{s_{n+1,2}}{x_{n+1,2}}+\cdots +\frac{s_{n+1,n}}{x_{n+1,n}} = 0.
\label{eq:scateqn}
\ee
One can easily express the $A_{a, a+1}$ variables in terms of $s_{n+1,a}$ as follows
\be
A_{12}=A_{n1}+s_{n+1,1},\quad A_{23}=A_{n1}+s_{n+1,1}+s_{n+1,2},\quad\ldots,\quad A_{n-1,n}=A_{n1}+\sum_{i=1}^{n-1}s_{n+1,i}.
\ee
Note that this requires $\sum_{i=1}^{n}s_{n+1,i}=0$.
In the context of CHY formalism \cite{Cachazo:2013hca} the condition \eqref{eq:scateqn} is precisely the scattering equation for label $n+1$.  For $m=n+1$ massless particles, the scattering equations relate punctures in $\mathbb{CP}^1$, located at $x_a$, $a=1,2,\ldots, m$, to kinematic invariants. The latter are nicely encoded in a symmetric matrix with components $s_{ab}$, whose diagonal elements are set to zero by the on-shell conditions. The space is $m(m{-}3)/2$ dimensional as one further imposes $m$ momentum conservation conditions $\sum_{b=1}^m s_{ab} = 0$, removing in total $2m$ variables from the $m(m{+}1)/2$ elements in a generic symmetric matrix.
The scattering equations then fix the positions of the punctures $x_{a}$ by imposing
\begin{eqnarray}
\sum_{\substack{a=1\\ a\neq b}}^{m-1} \frac{s_{ab}}{x_{a}-x_{b}}& =&0 , \,\qquad \, b=1,2,\ldots, m. 
\label{eq:scateq}
\end{eqnarray}

We have found that only the momentum conservation condition and the scattering equation corresponding to the $(n{+}1)^{\text{th}}$ particle is needed for the decoupling of such label. Given such conditions we would like to derive relations among Parke--Taylor factors from \eqref{fiden}. This requires us to find a way to ignore ${\cal O}(\Gamma_{n+1}^2)$. In order to achieve this it is necessary to introduce an expansion parameter. This is done by the replacement
$s_{ab}\to \epsilon\, s_{ab}$ and declaring that $A_{n1}=1$. Clearly $\Gamma_{n+1} = {\cal O}(\epsilon)$ and an expansion around $\epsilon=0$ of \eqref{fiden} leads to relations for the leading and next-to-leading order terms.

The leading order term gives $A_{ab}=1$ and becomes the $U(1)$ decoupling identity. The next to leading order term becomes 
\be
\left(\sum_{i=1}^{n-1}s_{n+1,i}\right){\rm PT}(1,2,\ldots,n{-}1,n{+}1,n)+\cdots +s_{n+1,1}{\rm PT}(1,n{+}1,2,\ldots ,n) =0,
\label{eq:bcjfund}
\ee
which is the fundamental BCJ relation \cite{Bern:2008qj} at the level of Parke--Taylor factors. In \cite{Cachazo:2012uq} this relation was in fact shown to hold on the support of the scattering equations.\footnote{Alternatively, CHY integrands can be understood as elements of the graded algebra $\mathsf{B}^\ast$ generated by $1$ and $\omega_{ab} := d\log x_{ab}$, isomorphic to an Orlik--Solomon algebra \cite{orlik2013arrangements}. The intersection pairing between two copies of $\text{H}^{n-3}(\mathsf{B}^\ast, \omega\wedge) \cong \text{H}^{n-3}(\mathcal{M}_{0,n},d{\pm}\omega \wedge)$ \cite{esnault1992cohomology} for $\omega := \sum_{a<b} s_{ab}\,\omega_{ab}$ computes CHY amplitudes \cite{Mizera:2017rqa}.}
 
As we extract higher orders in $\epsilon$, we will see that the higher powers of $\Gamma_{n+1}$ in the expansion \eqref{eq:expigamma} do not vanish. Instead, it turns out that only the vanishing of the linear term, imposed through the scattering equation \eqref{eq:scateqn}, is enough to obtain a new family of relations between the Parke--Taylor factors in \eqref{fiden} and multi-trace-like objects.

\subsubsection{Beyond BCJ Relations}

The derivation of the $U(1)$ decoupling identity and the fundamental BCJ relations presented in the previous section motivates the study of higher orders in the expansion. Here we will illustrate with this the second order identity, i.e., $\mathcal{O}(\epsilon\, ^2)$, obtained from the expansion, and we will leave the study of higher orders to future work.

The second order identity reads:
\begin{align}
&\left(\sum_{i=1}^{n-1}s_{n+1,i}^{2}\right){\rm PT}(1,2,\ldots,n+1,n) +\ldots + s_{n+1,1}^{2}{\rm PT}(1,n+1,2,\ldots,n) \nonumber\\
&\qquad = \Lambda_{1n}\sum_{1\leq p<q<n}
\left( {\textstyle\sum}_{r=p+1}^{q} s_{n+1,r}\right)^{\!2}\,
{\rm PT}(q+1,\ldots p,n+1)\, {\rm PT}(p+1,\ldots q,n+1)\,,
\label{eq:2bcj}
\end{align}
where the combination $\Lambda_{1n} := x_{n,n+1} x_{1,n+1}/x_{1,n}$ is known as the inverse soft factor. In the RHS we choose two arbitrary labels (in this case $1,n$) which are special in that they fix the summation range. The combination appearing in \eqref{eq:2bcj} exhibits a double trace structure. More explicitly, we have
\begin{align}
\mbox{{\rm PT}}(q{+}1,\ldots, p,n{+}1)\mbox{{\rm PT}}(p{+}1,\ldots, q,n{+}1) = \Lambda^{-1}_{p,q+1} \mbox{{\rm PT}}(q{+}1,\ldots, p)\mbox{{\rm PT}}(p{+}1,\ldots, q,n{+}1),
\end{align}
which turns the RHS of \eqref{eq:2bcj} into a sum of double-trace Parke--Taylor factors weighted by (inverse) soft factors. We leave the details of the second order computation for Appendix~\ref{app:PT-identities}. Here instead let us present a five-point example that can be explicitly checked. Taking $n{=}4$ produces:
\begin{align}
&s_{52}^2 \, \text{PT}(1, 5, 2, 3, 4) + (s_{51}{+}s_{52})^2\, \text{PT}(1, 2, 5, 3, 4) + (s_{51}{+}s_{52}{+}s_{53})^2 \, \text{PT}(1, 2, 3, 5, 4) \nonumber\\
&=\Lambda_{14} \Big( \Lambda^{-1}_{13} s_{52}^2 \text{PT}(2 5) \text{PT}(3, 4, 1) 
+\Lambda^{-1}_{14} (s_{52}{+}s_{53})^2 \text{PT}(2, 3, 5)\text{PT}(4, 1) \\
&\qquad\qquad\qquad\qquad\qquad\qquad\qquad\qquad\qquad\qquad + \Lambda^{-1}_{24} s_{53}^2 \text{PT}(3, 5) \text{PT}(4, 1, 2) \Big).\nonumber
\end{align}

Let us now provide an interpretation for these new identities. It is well-known that the fundamental BCJ relations exhausts the relations among Yang--Mills partial amplitudes. However, the string theory derivation of the identities which involves moving a vertex operator around the boundary of a disk \cite{BjerrumBohr:2009rd,BjerrumBohr:2010zs}, implies and infinite number of other identities involving higher-derivative terms in the EFT expansion of the open string. In string theory, the monodromy picked up by passing one vertex operator over another is an exponential of $\epsilon\, s_{ab}$, with $\epsilon$ identified with $\alpha'$. Expanding in $\epsilon$ gives the $U(1)$ decoupling and BCJ identities as the first two others. The ${\cal O}(\epsilon^2)$ contains amplitudes involving $\text{Tr}F^4$ terms, whose CHY formulation was found by He and Zhang in \cite{He:2016iqi}. Let us write down the monodromy relations for these CHY integrands, which hold on the support of the scattering equations:
\begin{align}
\left(\sum_{i=1}^{n-1}s_{n+1,i}^{2}\right)&{\rm PT}(1,2,\ldots,n{+}1,n)+ \ldots  +s_{n+1,1}^{2}{\rm PT}(1,n{+}1,2,\ldots,n)\nonumber\\
&= \sum_{1\leq p < q < r < s\leq n+1} x_{pq}\frac{s_{qr}}{x_{qr}} x_{rs} \frac{s_{sp}}{x_{sp}}\text{PT}(1,2,\ldots,n{+}1) + \ldots,
\end{align}
where the ellipsis on the RHS stands for all the insertions of particle label $n{+}1$ in the ordered set $(1,2,\ldots,n)$. Note that the cross ratio on the RHS is in fact a product of the two soft factors $(\Lambda^{(p)})^{-1}_{qr}\Lambda^{(r)}_{sp}$, with respect to labels $p$ and $r$ respectively.

We have seen that the $\Delta$-algebra parent formula provides a different representation of the monodromy identity. In particular, \eqref{eq:2bcj} exhibits the same soft behaviour on the particle $n{+}1$ on both sides of the equation. We leave the discussion of higher orders in $\epsilon\, s_{ab}$ as well as their explicit equivalence to monodromy relations for future work.

\subsection{KK Relations and Identities for Shuffle Sums}

There is another direction in which the $U(1)$ decoupling identity can be generalized. It is well-known that at the level of Parke--Taylor factors $U(1)$ decoupling is the simplest example of a more general set of relations known as KK identities. In its standard form, the KK relations are
\be
\text{PT}(1,\{\alpha\},n,\{\beta\})=(-1)^{|\beta|}\sum_{\omega\in \alpha \shuffle \beta^{T}}\text{PT}(1,\{\omega\},n).\label{kkrel}
\ee

The shuffle product $\alpha \shuffle \beta^{T}$ corresponds to all the permutations of the labels in $\{\alpha\} \cup \{\beta^T\}$ such that the respective orderings of $\alpha$ and $\beta^T$ (i.e., the transpose permutation of $\beta$) are preserved. One can now ask how the sum over the permutations in $\alpha \shuffle \beta^{T}$ is realized at the level of sums of triples. It is the goal of this section to provide a linear version of \eqref{kkrel}, as well as other useful identities at various orders in $\Delta$'s. The relations \eqref{kkrel} will then be a direct consequence of such parent identities. 

In order simplify the notation, let us
take $\alpha=(2,3,\ldots,m)$ and $\beta=(m{+}1,m{+}2,\ldots,n{-}1)$.
Such orderings correspond to the following object at the linear level:
\begin{align}
d:= d_{1,2,\ldots,m,n,m+1,\ldots,n-1} 
 =& \Delta_{123}+ \ldots+\Delta_{1mn}+\Delta_{1,n,m+1}+\ldots+\Delta_{1,n-2,n-1}\,.\label{eq:ptd}
\end{align}

We will also denote by $d_{\omega_i}$ the shuffles of $d$, in the case where $\omega_i$ is a permutation in $(2,\ldots,m)\shuffle(n{-}1,\ldots,m{+}1)$. In this setup, the realization of the KK relations \eqref{kkrel} is
\begin{equation}
\sum_{i=1}^{\binom{n-2}{m-1}}d_{\omega_{i}}^{n-2}=(-1)^{n-m-1}d^{n-2}.\label{eq:KKd}
\end{equation}

We are now in position to ask if there exists any identity at linear order between $\{d_{\omega_{i}}\}$
and $d$ giving rise to the physical KK relation, in an analogous way to the $U(1)$ and BCJ relation. To see this, let us split the RHS of \eqref{eq:ptd} as
\begin{equation}
\boxempty=\Delta_{123}+\mbox{\ensuremath{\ldots}+\ensuremath{\Delta}}_{1mn}\,\qquad\hat{\boxempty}=\Delta_{1,n,m+1}+\ldots+\Delta_{1,n-2,n-1}\,,
\end{equation}
i.e., $d=\boxempty+\hat{\boxempty}$. Note that these objects each
takes $(m-1)$ and $(n-m-1)$ triangles, respectively. It turns out
that for each shuffle $\omega$ one can define an induced splitting
of $d_{\omega}$, that is
\begin{equation}
d_{\omega}=\boxempty_{\omega}-\hat{\boxempty}=\boxempty-\hat{\boxempty}_{\omega}\,.\label{eq:ddecomp}
\end{equation}

We show in Appendix~\ref{app:shuffle} that each simplicial complex in $\boxempty_{\omega}$ (respectively $\hat{\boxempty}_{\omega}$)
is composed of $m{-}1$ (respectively $n{-}m{-}1$) triangles, and that we
have
\begin{equation}
\sum_{i=1}^{\binom{n-2}{m-1}}\boxempty_{\omega_{i}}=\binom{n-3}{m-2}\times d\,.\label{eq:linearbox}
\end{equation}
Furthermore, we also show that the following relations between powers
of $\boxempty_{\omega}$ hold:
\begin{align}
\sum_{i=1}^{\binom{n-2}{m-1}}\boxempty_{\omega_{i}}^{m-1} & = d^{m-1}\,,\label{eq:powerd}\qquad \sum_{i=1}^{\binom{n-2}{m-1}}\hat{\boxempty}_{\omega_{i}}^{n-m-1} = d{}^{n-m-1}\,.
\end{align}
Let us first focus on the linear relation \eqref{eq:linearbox}. Using the decomposition \eqref{eq:ddecomp} we find
\begin{equation}
\sum_{i=1}^{\binom{n-2}{m-1}}d_{\omega_{i}}=\binom{n-3}{m-2}\times\boxempty-\binom{n-3}{n-m-2}\times\hat{\boxempty}\,.\label{eq:linearkk}
\end{equation}
This is the reflection of the KK identity \eqref{eq:KKd} at the linear
order in $\Delta$. For instance, for $n=6$ and $m=3$ this identity reads:
\be
d_{123546}+d_{125346}+d_{152346}+d_{125436}+d_{152436}+d_{154236}=3\boxempty-3\hat{\boxempty},
\label{eq:kkex1}
\ee
where 
\be
\boxempty=\Delta_{123}+\Delta_{136}\,,\qquad\hat{\boxempty}=\Delta_{164}+\Delta_{145}\,.
\ee
The relative minus sign appearing
in the RHS of \eqref{eq:kkex1} is the reflection of a fact pointed
out in \cite{Arkani-Hamed:2014bca} at the level of
Parke--Taylor factors. There, the KK relation was derived by considering the triples 
\be
\{(123),(134),\ldots,(1,m{-}1,m),(1,m,n),(1,n,m{+}1),(1,m{+}1,m{+}2),\ldots,(1,n{-}2,n{-}1)\}
\ee
which pick up $\text{PT}(1,2,3,\ldots,m,n,m{+}1,\ldots,n{-}1)$ as the
only compatible one.\footnote{The ordering $\sigma$ is compatible with a set of triples ${\cal T}$ if each triple $(i_a,j_a,k_a)\in {\cal T}$ is a cyclic subword of $\sigma$. The decomposition property of the leading singularity of $\cal T$ into its compatible Parke--Taylor factors follows non-trivially from decomposing $G(2,n)$ into positive regions \cite{Arkani-Hamed:2014bca}.} This is the right hand side of \eqref{eq:KKd}. Changing the orientation of every triple after $(1,m,n)$ (which at the linear level corresponds to introducing the relative minus sign in \eqref{eq:linearkk}) leads to the set 
\be
\{(123),(134),\ldots,(1,m{-}1,m),(1,m,n),(n,1,m{+}1),(m{+}1,1,m{+}2),\ldots,(n{-}2,1,n{-}1)\}
\ee
This set has as compatible orderings the Parke--Taylor factors $\text{PT}(1,\{\omega_i\},n)$ in the left hand side of \eqref{eq:KKd} (which at the linear level correspond to the sum in \eqref{eq:linearkk}).

Finally, let us close this section by giving a direct derivation of the KK relations
at the level of Parke--Taylor factors, as following directly from their parent identities. We can multiply both sides of the first equation in \eqref{eq:powerd} by $\hat{\boxempty}^{n-m-1}$ to give:
\be
\sum_{i=1}^{\binom{n-2}{m-1}}\hat{\boxempty}^{n-m-1}\boxempty_{\omega_{i}}^{m-1}=\boxempty{}^{m-1}\hat{\boxempty}^{n-m-1}.
\ee

Here we have used that $\hat{\boxempty}$ is composed of $n-m-1$
triangles and hence $\hat{\boxempty}^{n-m}=0$. Finally,
we can use \eqref{eq:ddecomp} together with $\boxempty_{\omega_{i}}^{m}=0$
to write
\begin{align}
\sum_{i=1}^{\binom{n-2}{m-1}}d_{\omega_{i}}^{n-2} & = (-1)^{n-m-1}\binom{n-2}{n-m-1}\sum_{i=1}^{\binom{n-2}{m-1}}\hat{\boxempty}^{n-m-1}\boxempty_{\omega_{i}}^{m-1}\nonumber\\
& = (-1)^{n-m-1}\binom{n-2}{n-m-1}\boxempty{}^{m-1}\hat{\boxempty}^{n-m-1}\nonumber\\
& = (-1)^{n-m-1}d^{n-2}\,,
\end{align}
which is the identity \eqref{eq:KKd} we wanted to prove.

\subsection{\label{sec:more-general-identities}A Preview of More General Identities Among Leading Singularities}

The construction leading to the $U(1)$ decoupling identity from the parent \eqref{eq:lineardecoupling} can be generalized so as to obtain new relations among leading singularities. In fact the same parent identity could have become the equivalence of two distinct sets of $n{-}1$ triples by simply moving one $\Delta_{i,i+1,n+1}$ from one side of the equation to the other. More explicitly, by moving $\Delta_{n,1,n+1}$ in  \eqref{eq:lineardecoupling} to the LHS and then using the triangulation \eqref{eq:lineardecoupling2}, we find
\be
  \Delta_{1,2,n}+\Delta_{2,3,n}+\cdots + \Delta_{n-2,n-1,n}+\Delta_{n-1,1,n}  - \Delta_{n,1,n+1} = \Delta_{1,2,n+1}+\Delta_{2,3,n+1}+\cdots +\Delta_{n-1,n,n+1}.
\ee
This is nothing but two different triangulations of an $(n{+}1)$-gon. Computing the $(n{-}1)^{\text{th}}$ power on both sides leads to the same standard Parke--Taylor function ${\rm PT}(1,2,\ldots, n,n{+}1)$ as expected.

Consider now the octahedral leading singularity and try to obtain a new relation by reversing the previous procedure. The first step is to use the equivalence of its two presentations in triples,
\be
\Delta_{123}+\Delta_{345}+\Delta_{561}+\Delta_{642} = \Delta_{234} + \Delta_{456} + \Delta_{612} + \Delta_{153}.
\ee  
Next, a simple identity can be obtained by moving $\Delta_{234}$ from the RHS to the LHS and computing the fourth power of both sides:
\be
(\Delta_{123}+\Delta_{345}+\Delta_{561}+\Delta_{642} - \Delta_{234})^4= (\Delta_{456} + \Delta_{612} + \Delta_{153})^4.
\ee
Note that the RHS only has three $\Delta$'s and therefore vanishes leading to
\begin{align}
\Delta_{123}\Delta_{345}\Delta_{561}\Delta_{642} + \Delta_{234}\Delta_{345}\Delta_{561}\Delta_{642}+ \Delta_{123}\Delta_{234}\Delta_{561}\Delta_{642} & \nonumber\\
+\Delta_{123}\Delta_{345}\Delta_{234}\Delta_{642}
+\Delta_{123}\Delta_{345}\Delta_{561}\Delta_{234} &= 0.
\end{align}
The resulting identity expresses the octahedral leading singularity, first term in the sum, as a sum of four Parke--Taylor factors, each of them with a soft factor attached.

These two examples can be further generalized as follows. 

Consider any given set of $n{-}2$ triples ${\cal T}_1$ and find a second set of $m\geq n-2$  triples ${\cal T}_2$ such that 
\be
\sum_{\tau \in {\cal T}_1} \alpha_\tau \Delta_{\tau_1,\tau_2,\tau_3} = \sum_{\tau \in {\cal T}_2} \beta_\tau \Delta_{\tau_1,\tau_2,\tau_3}\,.
\ee 
Note that when ${\cal T}_1$ does not allow any moves $m$ is necessarily larger than $n{-}2$.

By taking the $(n-2)^{\text{th}}$ power of this parent identity or some rearrangement of it (obtained by moving terms from the RHS to LHS) we obtain new identities. Notice that when the LHS is taken to the $(n-2)^{\text{th}}$ power, one of the terms will always correspond to the leading singularity of ${\cal T}_1$.

Let us illustrate these general identities with an example involving a singularity that does not admit any moves. Consider the set of triples depicted in Figure \ref{fig:hourglass}:
\be
{\cal T}_1= \{(1,3,2),(1,4,7),(7,6,3),(4,5,6),(2,7,5)\}\,,\label{eq:hourglass}
\ee

\begin{figure}[h!]
	\centering
	\includegraphics[scale=0.9]{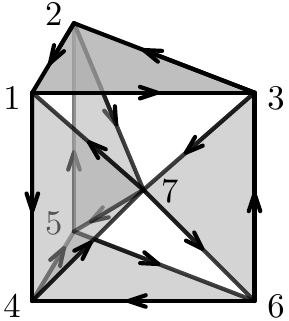}
	\caption{\label{fig:hourglass}The triangles and edges corresponding to \eqref{eq:hourglass}. They are not embeddable on the surface of a sphere.}
\end{figure}

Since this object is rigid it must be that we require $m>5$ to produce relations. It turns out that it is necessary to consider at least $m=7$ triples. The choice is not unique. For example, using 
\be
{\cal T}_2= \{(1,5,2),(1,4,5),(7,5,6),(7,6,3),(7,3,2),(4,7,6),(1,3,7)\}\,
\ee
it is easy to check that 
\be
\sum_{\tau \in {\cal T}_1}  \Delta_{\tau_1,\tau_2,\tau_3} =  \sum_{\tau \in {\cal T}_2}  \Delta_{\tau_1,\tau_2,\tau_3}\,.
\ee 
Taking the fifth power of LHS gives $F_{\mathcal{T}_1}$, whereas the fifth power of RHS leads to $\binom{9}{2}=36$ terms, each of them being a leading singularity. This $37$-term identity can be replaced by a $12$-term identity by simply moving $\Delta_{467}$ from the RHS to the LHS, i.e.
\be
\left(\sum_{\tau \in {\cal T}_1}  \Delta_{\tau_1,\tau_2,\tau_3} - \Delta_{4,6,7}\right) ^5 = \left( \sum_{\tau \in {\cal T}_2 \setminus \{(4,6,7)\}}  \Delta_{\tau_1,\tau_2,\tau_3}\right)^5\,.
\label{simsym}
\ee 
This time the fifth power only gives six terms on the LHS and five on the RHS (one of them is zero). It can be checked that each of the new terms simply corresponds to Parke-Taylor factors with soft factors attached. This gives a presentation of ${\cal T}_1$ in terms of $10$ simple leading singularities. Note that ${\cal T}_1$ is the unique, up to relabeling, irreducible leading singularity at seven points \cite{He:2018pue}.   

Let us finally point out that the (unoriented) set of triples in $\mathcal{T}_1$ contains a $G \simeq S_3\rtimes\mathbb{Z}_2$ symmetry. The $S_3$ piece acts by permuting the labels $\{1,2,3\}$ (and $\{4,5,6\}$ in the corresponding way), whereas $\mathbb{Z}_2$ flips top and bottom of Figure~\ref{fig:hourglass}.  The symmetry group is realized explicitly as a subgroup of $S_6$ in terms of generators as for example
\be
G = \langle (12)(46),(23)(56),(14)(25)(36)\rangle.
\ee
Even though $\mathcal{T}_1$ is invariant under these transformations, $\mathcal{T}_2$ is not. Therefore one finds different representations of the identity \eqref{simsym} given by the action of the group on the labels of $\mathcal{T}_2$. If we were to draw on-shell diagrams for the different sets of triples associated with $\mathcal{T}_2$ and its images we would then find diagrams related via moves. Of course, these on-shell diagrams do not have the correct dimension to be leading singularities. We leave the study of such moves for future work. 

\section{\label{sec:higher-k}N$^{k-2}$MHV Leading Singularities and the $\Delta$-Algebra}

In this section we turn to the generalization of the $\Delta$-algebra to helicity sectors other than MHV. A general on-shell diagram describing a leading singularity in the N$^{k-2}$MHV sector must satisfy the following condition $\alpha+\beta = 2(k{-}2)$, where $\alpha$ is the number of white vertices that are not connected to any external label either directly or via a path only involving white vertices, while $\beta$ is the number of edges that connect two black vertices \cite{Arkani-Hamed:2014bca}. It is now clear why the $k=2$ sector is special. All MHV diagrams must have $\alpha=\beta=0$ and this leads to the description in terms of lists of $n{-}2$ triples of labels. Already the $k=3$ sector allows for several possibilities, since $\alpha+\beta=2$.

One attempt to simplify the general problem is to restrict to diagrams with $\alpha=0$. In this case one can contract all clusters of trivalent black vertices into higher-valent black vertices and associate with them a tuple of labels. For example, if $k=3$ then $\beta=2$ and there are two possibilities. The first is to have two $4$-valent black vertices and $n{-}5$ trivalent black vertices. The second is to have one $5$-valent black vertex and $n{-}4$ trivalent black vertices.

Having sets of labels of different length seems to prevent the development of a $\Delta$-algebra formulation. Luckily, Franco, Galloni, Penante, and Wen \cite{Franco:2015rma} showed that at the level of the integrand in the Grassmannian formulation of scattering amplitudes \cite{ArkaniHamed:2012nw} it is possible to work with only $(k{+}1)$-tuples of labels and lists containing exactly $n{-}k$ of them. The price to pay is to abandon a direct connection to on-shell diagrams. The reason is that one has to work in the top cell of $G(k,n)$ which has dimension $k(n{-}k)$, while leading singularity on-shell diagrams live in $2(n{-}2)$-dimensional cells. However, after computing with $(k{+}1)$-tuples, one can go to boundaries of $G(k,n)$ of dimension $2(n{-}2)$ and make connection to leading singularities. We will not discuss this last step and instead concentrate on the $\Delta$-algebra related to the integrand.

The two-components objects $\lambda_a$ are replaced by $k$-component ones $Z_a \in \mathbb{P}^{k-1}$ and Pl\"ucker coordinates $\langle a,b\rangle $ are replaced by $\la 1,2, \ldots, k \ra := \det(Z_1, Z_2, \ldots, Z_{k})$.  
Given a list of $(k{+}1)$-tuples
\be
{\cal T}=\{ (\tau^{(1)}_1,\tau^{(1)}_2,\ldots , \tau^{(1)}_{k+1}), \ldots ,(\tau^{(n-k)}_1,\tau^{(n-k)}_2,\ldots , \tau^{(n-k)}_{k+1})\}
\ee
the determinant formula \eqref{rtnl} can be generalized to \cite{Franco:2015rma}:
\be
F_{\cal T} := \frac{\left({\rm det}'M\right)^k}{\prod_{i=1}^{n-k}|\tau^{(i)}_1,\tau^{(i)}_2,\ldots , \tau^{(i)}_{k+1}|},
\ee
where 
\be
|a_1,a_2,\ldots ,a_{k+1}| = \la a_1,a_2, \ldots, a_{k} \ra\la a_2,a_3, \ldots, a_{k+1} \ra\cdots \la a_{k+1},a_1, \ldots, a_{k-1} \ra. 
\ee
Also the matrix $M$ is the natural generalization of the $k=2$ case. It is a $(n{-}k)\times n$ matrix with columns labeled by the particle number and rows by $(k{+}1)$-tuples in the list ${\cal T}$. The only non-zero entries in the row given by the $i^{\rm th}$ tuple are those labeled by elements in this tuple. The non-zero entries have values equal to the Pl\"ucker coordinate with labels other than the column in consideration. Finally, the reduced determinant is, up to a sign,
\be
{\rm det}'M := \frac{{\rm det}M^{(a_1,a_2,\ldots ,a_{k})}}{\la a_1,a_2, \ldots, a_{k} \ra}
\ee
and is independent of the choice of labels $a_i$ \cite{Franco:2015rma}.

Following the same steps that led to the definition of $\Delta$ for $k=2$, i.e., rewriting each of the $k$ powers of ${\rm det}'M$ in terms of $k$ kinds of Grassmann variables $\theta^{(r)}_i$, leads to
\be
\Delta_{1,2,\ldots ,k+1}:=\frac{\prod_{r=1}^k\left(\theta_1^{(r)}\la 2,3,\ldots ,k{+}1\ra - \theta_2^{(r)}\la 1,3,\ldots ,k{+}1\ra + \ldots + (-1)^{k} \theta_{k+1}^{(r)}\la 1,2,\ldots ,k\ra\right)}{|1,2,\ldots ,k{+}1|}.
\ee
In the numerator, terms accompanying $\theta_{i}^{(r)}$ come with a sign $(-1)^{i+1}$. The whole object picks up a minus sign under exchange of two adjacent labels, $\Delta_{\ldots ij \ldots} = -\Delta_{\ldots ji\ldots}$, as well as a sign $(-1)^k$ under cyclic shifts of labels.

Let us denote the product of $\Delta$'s in the set ${\cal T}$ by
\be
{\rm LS}_{\cal T}:= \prod_{i=1}^{n-k}\Delta_{\tau^{(i)}_1,\tau^{(i)}_2,\ldots ,\tau^{(i)}_{k+1}},\label{LS-higher-k}
\ee
so that
\be
F_{\cal T} = \int \prod_{r=1}^{k} \prod_{i=1}^{n} d\theta^{(r)}_i \,\text{LS}_{\cal T} \prod_{r=1}^{k} \frac{\theta^{(r)}_{a_{1}}\theta^{(r)}_{a_{2}}\cdots \theta^{(r)}_{a_{k}} }{\langle a_{1}, a_{2}, \ldots, a_{k} \rangle}
\ee
for an arbitrary choice of labels $a_{i}$. Of course, this is an abuse of notation, since in order to obtain a physical leading singularity associated with an on-shell diagram one has to not only perform the integration over $\theta_i^{(r)}$'s, but also reduce the dimension of the cell in $G(k,n)$ down to $2(n{-}2)$.

Notice that $\Delta$'s defined above have Grassmann degree $k$. This means that on top of squaring to zero property,
\be
\Delta_{a_1,a_2,\ldots,a_{k+1}}^2 = 0,
\ee
we also have the (anti)commutation relations:
\begin{align}
\left[\Delta_{a_1,a_2,\ldots, a_{k+1}} , \Delta_{b_1,b_2,\ldots, b_{k+1}}\right] = 0 &\qquad\text{for } k \text{ even},\nonumber\\
\{\Delta_{a_1,a_2,\ldots, a_{k+1}} , \Delta_{b_1,b_2,\ldots, b_{k+1}}\} = 0 &\qquad\text{for } k \text{ odd}.
\end{align}
Hence, following the steps in the MHV case, in the $k$ even case we can rewrite \eqref{LS-higher-k} as
\be
{\rm LS}_{\cal T} = \frac{1}{(n-k)!}\left(\sum_{i=1}^{n-k}\Delta_{\tau^{(i)}_1,\tau^{(i)}_2,\ldots ,\tau^{(i)}_{k+1}}\right)^{n-k}.
\ee
Due to anticommutation of $\Delta$'s, there is no natural analogue of this expression for $k$ odd, though one might still construct $\text{LS}_{\cal T}$ in interesting ways out of products of linear combinations of $\Delta$'s. Moreover, one can verify that $\Delta_{a_1,a_2,\ldots ,a_{k+1}}$ satisfy  
\be
\Delta_{1,2,\ldots,k+1} = \sum_{i=1}^{k+1} \Delta_{1,2,\ldots, i{-}1, R, i{+}1, \ldots, k{+}1}\label{higher-k-delta}
\ee
for an arbitrary $Z_R$. This relation is in fact the so-called $1{-}(k{+}1)$ bistellar flip, or Pachner move \cite{PACHNER1991129}. This leads to the generalization of $u_{ij}$ and to a formula in terms of ``facet variables'' $u_{a_1,a_2,\ldots,a_k}$ such that 
\be
\Delta_{1,2,\ldots ,k+1} = u_{1,2,\ldots,k}+ (-1)^k u_{2,3,\ldots ,k+1}+ u_{3,4,\ldots,k+1,1} + \ldots + (-1)^k u_{k+1,1,\ldots ,k-1}.\label{higher-k-u}
\ee

The identity \eqref{higher-k-delta} involves $k{+}2$ different labels. Since computing leading singularities, according to \eqref{LS-higher-k}, requires the $(n{-}k)^{\text{th}}$ power of $\Delta$'s, for $n = k{+}2$ we are always interested in quadratic identities. The most natural one is obtained by multiplying both sides of \eqref{higher-k-delta} by $\Delta_{1,2,\ldots,k+1}$ from the left, giving:
\be
\Delta_{1,2,\ldots,k+1} \sum_{i=1}^{k+1} \Delta_{1,2,\ldots, i{-}1, k+2, i{+}1, \ldots, k{+}1} = 0.
\ee
Here we set $R=k{+}2$ for simplicity. This is the analogue of the $U(1)$ decoupling identity at $k{=}2$.

Other types of identities can be obtained by moving $\Delta$'s from RHS to LHS of \eqref{higher-k-delta} and squaring both sides (for even $k$), or alternatively multiplying the whole equation by one of the sides. It turns out that $k{=}2$ is the only case for which this operation relates a single leading singularity to another one. This is because only for $k{=}2$ the $\Delta$'s can be partitioned into two on the LHS and two on the RHS, which after squaring gives an equality between two leading singularities: the square move.

For instance, \eqref{higher-k-delta} for $k=4$ implies that the $\Delta$'s satisfy the identity (the $3{-}3$ bistellar flip):
\be\label{rinv}
\Delta_{12345} + \Delta_{26345} + \Delta_{61345} = \Delta_{12645} + \Delta_{12365} + \Delta_{12346},
\ee
where we also used antisymmetry under exchanging labels, $\Delta_{abcde} = -\Delta_{bacde}$.
This identity is very familiar in the context of NMHV amplitudes in momentum twistor space. In fact, identifying our variables $Z_a$ with momentum twistors and $\theta^{(r)}_a$ with the $Z_a$ superpartner $\eta^I_a$, one finds that $\Delta_{abcde}$ is exactly what is known as the $R$-invariant $R_{abcde}$. Moreover, \eqref{rinv} is nothing but the equivalence of the six-point NMHV amplitude when computed using the BCFW technique or its parity conjugated version.    

The straightforward analogy with $R$-invariants somehow breaks down when we square both sides of  \eqref{rinv} in order to get a relation among $k=4$ leading singularities.\footnote{Recall that by this we mean the integrand of the Grassmannian integral formulation of the object.} The identity is then
\be
\Delta_{12345}\Delta_{26345} + \Delta_{12345}\Delta_{61345} + \Delta_{26345}\Delta_{61345} = \Delta_{12645}\Delta_{12365} + \Delta_{12645}\Delta_{12346} + \Delta_{12365}\Delta_{12346}
\ee
and it is the analogue of the square move identity $\Delta_{123}\Delta_{134} = \Delta_{234} \Delta_{241}$ at $k=2$. Similarly, squaring both sides of \eqref{higher-k-u} reveals that $u$'s satisfy
\be
u_{1234} ( u_{2345} + u_{3451}) + \text{cycl.}=0,
\ee
which is the analogue of the identity $u_{12}u_{23} + \text{cycl.} = 0$ at $k=2$.

It is clear that there is much more to be explored in the higher-$k$ extension of the $\Delta$-algebra. A particularly interesting direction is the recent use of volumes of $k$-simplices in the computation of $\phi^3$ biadjoint scalar amplitudes in the context of their corresponding amplituhedron, i.e., the associahedron in the Mandelstam space \cite{Arkani-Hamed:2017mur,He:2018svj}. We leave the exploration of consequences of the above algebras for future work.

\section{\label{sec:discussion}Discussions and Future Directions}

In this work we introduced the $\Delta$-algebra as a natural structure in the construction of MHV non-planar leading singularities. Physical applications of the algebra structure range from manifesting symmetries, such as cyclic invariance in the Parke--Taylor case, to revealing new equivalences of on-shell diagrams not connected via square moves.  Also, identities involving several leading singularities found a natural description in terms of parent (or linear) relations in the algebra. It is clear that we only scratched the surface of the structure. In particular, extensions of identities to $k>2$ and applications to loop integrands, as well as less supersymmetric \cite{Benincasa:2015zna} and gravity on-shell diagrams \cite{Heslop:2016plj,Herrmann:2016qea} are some of the most pressing directions. Another direction which is hard to overlook is to find possible connections to the triangulation of the $m=2$ amplituhedron for general $k$ as discussed in \cite{Arkani-Hamed:2014dca,Arkani-Hamed:2017tmz} where a structure of the form ``$({\rm polygon})^k$" was found. 

In addition, there are some directions for future developments for which we can provide more detailed explanations.

\subsection{Comparison with (Combinatorial) BCFW Relations}

It is well-known that it is possible to obtain relations among leading singularities by attaching a BCFW-bridge to one of them and using the residue theorem to get the others \cite{Britto:2004ap,Britto:2005fq}. A BCFW-bridge is a 4-particle diagram with a single black vertex connected to a single white vertex (see Chapter~17 of \cite{ArkaniHamed:2012nw} for a review). It is not a leading singularity but it is useful when two of the four legs are connected to a complicated leading singularity. In fact, it can be shown that this procedure decomposes any MHV leading singularity into a sum of leading singularities with at least one soft factor attached. Here we compare the identities we obtained by moving $\Delta$'s in Section~\ref{sec:more-general-identities} with those from a purely combinatorial version of the BCFW procedure when applied to a list of triples ${\cal T}$.

Let us start by explaining the {\it combinatorial} BCFW identity related to a BCFW-bridge $(1,n)$.

Given a list of triples ${\cal T}$, identify all those where label $n$ enters and make a list ${\cal T}_n$. For each triple in ${\cal T}_n$, say $(nab)$, create the following two lists:
\be\label{exchange}
{\cal T}_{na} := \{(1na)\}\cup \Big({\cal T}-\{(nab)\}\Big)\Big|_{n\to a}
\ee
and
\be
{\cal T}_{nb} := \{(1bn)\}\cup \Big({\cal T}-\{(nab)\}\Big)\Big|_{n\to b}.
\ee
Note that the orientation of the triples $(1na)$ and $(1bn)$ was chosen to match that of the shared edges with the original triple $(nab)$. Also, when either $a{=}1$ or $b{=}1$, the corresponding contribution is defined to vanish, i.e., ${\cal T}_{n1} := 0$. The new lists either correspond to new leading singularities or to zero.

The combinatorial BCFW identity is then
\be
{\rm LS}_{\cal T} = \sum_{(nab)\in {\cal T}_n} \Big( {\rm LS}_{{\cal T}_{na}}+{\rm LS}_{{\cal T}_{nb}} \Big).\label{BCFW-sum}
\ee

The simplest example is the four-particle $U(1)$ decoupling identity. In our language we start with
\be
\Delta_{123} = \Delta_{124}+ \Delta_{234} + \Delta_{314}.
\ee
Multiplying both sides by $\Delta_{123}$ we get
\be\label{ude}
\Delta_{123}\Delta_{124} + \Delta_{123}\Delta_{234} + \Delta_{123}\Delta_{314} = 0.
\ee
In the BCFW framework, we can attached a BCFW-bridge $(2,4)$ to ${\rm PT}(1234)$. Following the combinatorial procedure lead us to \eqref{ude}.

This coincidence of results turns out to be an accidental coincidence. One can check that there is not any single BCFW-bridge that can be attached to the leading singularity $\{(123),(345),(561),(642)\}$ that gives rise to the identity found in Section~\ref{sec:more-general-identities} by moving a triangle, i.e.,
\begin{align}
\Delta_{123}\Delta_{345}\Delta_{561}\Delta_{642} = \Delta_{234}\Delta_{345}\Delta_{561}\Delta_{642}+ \Delta_{123}\Delta_{234}\Delta_{561}\Delta_{642} & \nonumber\\
+\Delta_{123}\Delta_{345}\Delta_{234}\Delta_{642}
+\Delta_{123}\Delta_{345}\Delta_{561}\Delta_{234} &.
\end{align}
For example, applying a BCFW-bridge $(3,6)$ to $\{(123),(345),(561),(642)\}$ one finds that
\begin{align}
\Delta_{123}\Delta_{345}\Delta_{561}\Delta_{642} = \Delta_{123}\Delta_{142}\Delta_{345}\Delta_{361} +
\Delta_{123}\Delta_{345}\Delta_{356}\Delta_{542} &\nonumber\\
+\Delta_{123}\Delta_{345}\Delta_{364}\Delta_{541} +
\Delta_{123}\Delta_{326}\Delta_{345}\Delta_{521} &.
\end{align}
Here we used that ${\cal T}_{6} = \{(561),(642)\}$ and hence \eqref{BCFW-sum} involves a sum over $(a,b) = (1,5), (4,2)$ with the corresponding lists
\be
{\cal T}_{61} = \{ (123), (142), (345), (361)\},\qquad
{\cal T}_{65} = \{ (123), (345), (356), (542) \},\\
\ee
as well as
\be
{\cal T}_{64} = \{ (123), (345), (364), (541) \},\qquad
{\cal T}_{62} = \{ (123), (326), (345), (521) \}.
\ee

It is clearly a very important question to find out how to connect these two constructions. Moreover, it is hard to overlook the fact that the combinatorial BCFW relation involves a procedure very reminiscent of the circuit elimination axiom of a matroid \cite{reiner2005lectures}.

\subsection{Gauge Redundancies in $\Delta_{abc}$ and Differential Form Interpretation}

The realization of the $\Delta$-algebra discussed in this paper is based on that for the building block 
\be
u_{ab} = \frac{(\theta_a-\theta_b)(\chi_a-\chi_b)}{x_{a}-x_{b}}.
\ee
However, this representation is not unique, there are redundancies similar to gauge transformations, for example,
\be
u_{ab} \to u_{ab} + \alpha \theta_a\theta_b
\ee
for any constant $\alpha$ leaves all properties of the algebra invariant. Under this transformation 
\be
\Delta_{abc}\to \Delta_{abc} + \alpha \left( \theta_a\theta_b+\theta_b\theta_c+\theta_c\theta_a\right).
\ee
Now, if both $u_{ab}$ and $\Delta_{abc}$ are not gauge invariant, what is the physical observable? The answer is, of course, the rational function associated to a leading singularity. More precisely,
\be
F_{\cal T} = \int \prod_{a=1}^n d\theta_a d\chi_a \, \prod_{i=1}^{n-2}\Delta_{a_ib_ic_i} \times \frac{\theta_d\theta_e}{x_{de}}\frac{\chi_f\chi_g}{x_{fg}}.
\ee
The fact that there are gauge redundancies suggests that an interpretation as differential forms could be the link to an even more geometrical picture. In fact, as mentioned in the introduction, the object $\Delta_{abc}$ defined in homogeneous variables becomes the $2$-form $\Omega_{abc}$ under the Grassmann variable $\leftrightarrow$ ($1$-form) identification used by He and Zhang \cite{He:2018okq}. It is interesting that when $\Omega_{abc}$ is written in inhomogeneous variables the form that appears is a gauge transform of $\Delta_{abc}(x)$. Let us show this more explicitly. 

Let us parametrize holomorphic spinors in the following way:
\be
\lambda_{i} = e^{\tilde{x}_i} \begin{pmatrix} 1 \\ x_i \end{pmatrix}, \qquad \text{so that} \qquad \langle ij \rangle = \epsilon_{\alpha\beta} \lambda_i^\alpha \lambda_j^\beta = -e^{\tilde{x}_i + \tilde{x}_j} x_{ij}.
\ee
Let us apply this to $\Omega_{ijk}$, which we rewrite below for the reader's convenience,
\be
\Omega_{ijk} = d\log \frac{\la ij \ra }{\la ki \ra} \wedge d\log \frac{\la jk \ra }{\la ki \ra} = \frac{\left( d\lambda^{1}_i \la jk \ra + \text{cycl.} \right)\wedge\left( d\lambda^{2}_i \la jk \ra + \text{cycl.} \right)}{\la ij \ra \la jk \ra \la ki \ra}.
\ee
Using the above change of variables we find that it decomposes as:
\be
\Omega_{ijk} = U_{ij} + U_{jk} + U_{ki}, \qquad \text{with}\qquad U_{ij} = d\tilde{x}_i \wedge d\tilde{x}_j - \frac{dx_{ij} \wedge d\tilde{x}_{ij}}{x_{ij}}.
\ee
Under the identification $(\theta_a,\chi_a)\to (dx_a,d\tilde x_a)$ we find a gauge transform version, with $\alpha=1$, of the representation we used. Now we are guaranteed that the 2-forms defined above also furnish a representation of $\text{H}^\ast\text{Conf}_{n}(\mathbb{R}^3)$, i.e.,
\be
U_{ij} = - U_{ji}, \qquad U_{ij} \wedge U_{ij} = 0, \qquad U_{ij} \wedge U_{jk} + \text{cycl.} = 0.
\ee
Note that if we were to identify $(x_a,\tilde x_a)\to (z_a,{\bar z}_a)$, then $U_{ij}$ would be a combination of a $(1,1)$-form and a $(0,2)$-form. However, the $(0,2)$ part is pure gauge. It would be interesting to explore this further, especially in view of the following discussion.

\subsection{${\rm det}'\Phi(x,y)$ and $u_{ij}$}

There is an object that made its appearance for the first time in a twistor formulation of gravity amplitudes introduced by one of the authors and Geyer in \cite{Cachazo:2012da} and which admits an interesting formulation in terms of $u_{ij}$. This is the reduced determinant of
\be
\Phi_{ab}(x) = \left\{
              \begin{array}{ll}
                \frac{s_{ab}}{x_{ab}^2}, & \hbox{if $a\neq b$;} \\
                -\sum_{c=1}^n\frac{s_{ac}}{x_{ac}^2}, & \hbox{if $a=b$.}
              \end{array}
            \right.
\ee
It is easy to show that if $x$'s are any solution to the scattering equations then
\be
\sum_{a=1}^n x_a^m\Phi_{ab}(x) = 0 \quad {\rm for} \quad m=\{0,1,2\}.
\ee
This implies that $\Phi$ is a $n\times n$ symmetric matrix of corank $3$.

The object of interest for this part is a generalization of $\Phi$ introduced in \cite{Cachazo:2013gna} to prove a property known as KLT orthogonality. The object is defined using the matrix
\be
\Phi_{ab}(x,y) = \left\{
              \begin{array}{ll}
                \frac{s_{ab}}{x_{ab}y_{ab}}, & \hbox{if $a\neq b$;} \\
                -\sum_{c=1}^n\frac{s_{ac}}{x_{ac}y_{ab}}, & \hbox{if $a=b$.}
              \end{array}
            \right.
\ee
Of course, $\Phi_{ab}(x,x)$ reduces to the previous definition. It turns out that if $x$ and $y$ are distinct solutions to the scattering equations then $\Phi_{ab}(x,y)$ has corank $4$ since
\be
\sum_{a=1}^n x_a^{m_1}y_b^{m_2}\Phi_{ab}(x) = 0 \quad {\rm for} \quad (m_1,m_2)\in \{(0,0),(1,0),(0,1),(1,1)\}.
\ee
In order to obtain a non-zero determinant it is necessary to remove four rows and four columns. However, for our purposes it will be sufficient to consider the submatrix of $\Phi(x,y)$ obtained by deleting rows $i,j,k$ and columns $p,q,r$. Clearly, the determinant vanishes unless $x_a$ and $y_a$ are the same solution to the scattering equations. Let us denote by $I_x$ and $I_y$ the label in $\{1,2,\ldots , (n{-}3)!\}$ corresponding to the solutions $x$ and $y$ belong to. Therefore one finds that the reduced determinant is
\be\label{ortho}
{\rm det}' \Phi(x,y) = \frac{{\rm det}\Phi_{ijk}^{pqr}(x,y)}{x_{ij}x_{jk}x_{ki}y_{pq}y_{qr}y_{rp}} = \delta_{I_x,I_y}{\rm det}' \Phi(x,x).
\ee

The Grassmann integral formulation is
\be\label{grassD}
{\rm det}' \Phi(x,y) =\int \prod_{m=1}^n(d\theta_m\,d\chi_m) \exp \left(\sum_{a,b}\theta_a\Phi_{ab}(x,y)\chi_b\right)\frac{\theta_i\theta_j\theta_j}{x_{ij}x_{jk}x_{ki}}
\frac{\chi_p\chi_q\chi_r}{y_{pq}y_{qr}y_{rp}}.
\ee
Once again it is possible to rewrite this in term of $u_{ij}$ as follows
\be\label{redDet}
{\rm det}' \Phi(x,y) =\int \prod_{m=1}^n(d\theta_m\,d\chi_m)\frac{\theta_i\theta_j\theta_j}{x_{ij}x_{jk}x_{ki}}
\frac{\chi_p\chi_q\chi_r}{y_{pq}y_{qr}y_{rp}}\times \left(\sum_{\substack{a,b=1\\ b\neq a}}^n\frac{s_{ab}}{y_{ab}}u_{ab}\right)^{n-3}
\ee
up to an irrelevant numerical factor.

Following the spirit of this work in the way leading singularities were associated with the product of $\Delta$'s it is natural to propose that ${\rm det}' \Phi(x,y) $ must likewise be associated with
\be
\left(\sum_{\substack{a,b=1\\ b\neq a}}^n\frac{s_{ab}}{y_{ab}}u_{ab}\right)^{n-3}.
\ee
Moreover, as proven by the second author in \cite{Early:2018mac}, any linear combination of $u_{ab}$'s with antisymmetric coefficients $c_{ab}$ that satisfy $\sum_a c_{ab}=0$ can be written in terms of $\Delta_{abc}$. Identifying $c_{ab}$ with $s_{ab}/y_{ab}$ and recalling that the $y$'s satisfy the scattering equations one finds
\be
\sum_{\substack{a,b=1\\ b\neq a}}^n\frac{s_{ab}}{y_{ab}}u_{ab} = \sum_{\tau\in {\cal T}}\alpha_\tau(s_{ab},y_{ab})\, \Delta_{\tau_1,\tau_2,\tau_3}
\ee
for some set of triples ${\cal T}$ and coefficients $\alpha_{\tau}$. It would be interesting to explore this connection further.

\section*{Acknowledgements}
We thank K.~Yeats and S.~Yusim for useful discussions. N.E. is grateful to A.~Ocneanu for many illuminating discussions during his graduate study about permutohedral plates and blades and related topics, and also to A.~Postnikov for discussions during his course on cluster algebras in Fall, 2018 at MIT. N.E. also would like to thank Perimeter Institute for their support and hospitality while this work was initiated. A.G. thanks CONICYT for financial support. This research was supported in part by Perimeter Institute for Theoretical Physics. Research at Perimeter Institute is supported by the Government of Canada through the Department of Innovation, Science and Economic Development Canada and by the Province of Ontario through the Ministry of Research, Innovation and Science.

\renewcommand{\thefigure}{\thesection.\arabic{figure}}
\renewcommand{\thetable}{\thesection.\arabic{table}}
\appendix

\section{\label{app:gluing}Reducible Leading Singularities}

Recall from \eqref{LS-map} that an $n$-point MHV leading singularity has the associated rational function $F_{\cal T}$ computed as follows:
\be\label{rational-fn}
F_{\cal T} = \int \prod_{i=1}^{n} d\theta_i d\chi_i \; \text{LS}_{\cal T}\; \frac{\theta_a \theta_b}{x_{ab}} \frac{\chi_c \chi_d}{x_{cd}},
\ee
where the labels $a,b,c,d$ can be chosen arbitrarily.

We can ask what happens if ${\cal T}$ contains a subset of triples $\widetilde{\cal T} \subset {\cal T}$, which itself forms a leading singularity, i.e.,
\be\label{LS-decomp}
\text{LS}_{\cal T} = \text{LS}_{\widetilde{\cal T}} \prod_{\tau \in {\cal T} \setminus \widetilde{\cal T}} \Delta_{\tau_1, \tau_2, \tau_3}.
\ee
Note that here the complementary set ${\cal T} \setminus \widetilde{\cal T}$ does not necessarily give a leading singularity on its own. Let us say, without loss of generality, that $\widetilde{\cal T}$ contains $m{-}2$ triples constructed out of the labels $\{1,2,\ldots,m\}$ for $m < n$. Plugging the decomposition \eqref{LS-decomp} back into the expression for the rational function \eqref{rational-fn}, we find:
\be\label{LS-prime-factor}
F_{\cal T} = \underbrace{\left( \int \prod_{i=1}^{m} d\theta_i d\chi_i \; \text{LS}_{\widetilde{\cal T}}\; \frac{\theta_a \theta_b}{x_{ab}} \frac{\chi_c \chi_d}{x_{cd}} \right)}_{F_{\widetilde{\cal T}}} \left( \int \!\!\!\prod_{j=m+1}^{n} d\theta_j d\chi_j \prod_{\tau \in {\cal T} \setminus \widetilde{\cal T}} \Delta_{\tau_1, \tau_2, \tau_3} \,\bigg|_{\substack{\theta_i = \chi_i = 0\\ i=1,2,\ldots,m}} \right),
\ee
since Grassmann integrals associated to the labels $\{1,2,\ldots,m\}$ are saturated in the left factor. We also chose the arbitrary labels $a,b,c,d$ to be taken out of the same set. The rational function factors and the left factor is precisely $F_{\widetilde{\cal T}}$. In the right factor we set the Grassmann variables from the set $\{1,2,\ldots,m\}$ to zero. Since this statement holds on the level of the rational functions, it is clearly independent of a given representation in terms of sets of triples. This extends the notion of reducibility of \cite{He:2018pue} and motivates the following definition.

\begin{definition}
A set ${\cal T}$ of $n{-}2$ triples is called \emph{reducible} if there exists a proper non-empty subset $\widetilde {\mathcal T} \subsetneq {\mathcal T}$ of $m{-}2$ triples consisting of exactly $m>3$ labels.
\end{definition}

Thus, as above, the rational function $F_\mathcal{T}$ contains $F_{\widetilde{\mathcal{T}}}$ as a factor. Here we excluded the trivial case where $m=3$. Note that by the same logic, if there exists a subset of $m{-}2$ triples $\widetilde{\cal T} \subset {\cal T}$, which contains \emph{fewer} than $m$ labels, then $\text{LS}_{\widetilde{\cal T}} = 0$ and hence also $\text{LS}_{\cal T}$ vanishes. For example, any leading singularity containing two repeated triples $\widetilde{\cal T}=\{ (a,b,c), (a,b,c)\}$ vanishes. 

In the following discussion we give examples of how to construct reducible leading singularities by gluing other ones.

\subsection{Soft Factor}

First, let us consider the simplest case in which $m=n{-}1$, i.e., removing one triple from a leading singularity leaves us with another leading singularity. The complementary set is necessarily of the form
\be
{\cal T} \setminus \widetilde{\cal T} = \{(y,z,n)\}, \qquad \text{where} \qquad y,z \in \{ 1,2,\ldots, n{-}1\}.
\ee
Hence the right factor in \eqref{LS-prime-factor} becomes:
\be
\int d\theta_n d\chi_n\, \Delta_{yzn} \bigg|_{\substack{\theta_y = \chi_y=0\\ \theta_z = \chi_z = 0}} = \int  d\theta_n d\chi_n \left( \frac{\theta_n \chi_n}{x_{zn}} + \frac{\theta_n \chi_n}{x_{ny}} \right) = \frac{x_{zy}}{x_{zn} x_{ny}}.
\ee
This is nothing but the so-called \emph{soft factor} \cite{ArkaniHamed:2009dn} glued onto the leading singularity $\text{LS}_{\cal T}$ by the legs $y$ and $z$.

\subsection{Gluing by an Edge}

The above soft factor can be alternatively understood as gluing of a single triple, say, $\{(n{+}1,n{+}2,n)\}$, to the original set of triple $\cal T$ by identifying $n{+}1$ with $y$ and $n{+}2$ with $z$. Let us use the notation:
\be
{\cal T} \overset{(y,z)}{\underset{(n+1,n+2)}{\circ}} \{ (n{+}1,n{+}2,n)\}
\ee
to define the result of this operation.

It is in fact the simplest example of a more general gluing operation, which takes two sets of triples ${\cal T}_1$ and ${\cal T}_2$ with $n_1{-}2$ and $n_2{-}2$ elements respectively, and identifies a pair of labels $(y,z)$ from the first set with a pair of labels $(v,w)$ from the second. The result is denoted by:
\be\label{gluing-by-edge}
{\cal T}_3 = {\cal T}_1 \overset{(y,z)}{\underset{(v,w)}{\circ}} {\cal T}_2.
\ee
It consists of $n_1 {+} n_2 {-} 4$ triples and involves exactly $n_1 {+} n_2 {-} 2$ labels. Therefore it is a valid leading singularity.

For the above gluing operation, the right factor in \eqref{LS-prime-factor} simplifies as follows:
\be
\int \prod_{\substack{j=n_1+1\\ j \neq v,w}}^{n_1 + n_2} d\theta_{j} d\chi_{j} \prod_{\tau \in {\cal T}_2} \Delta_{\tau_1,\tau_2,\tau_3} \bigg|_{\substack{v=y\\ w=z}} = x_{yz}^2\, F_{{\cal T}_2} \bigg|_{\substack{v=y\\ w=z}},
\label{gluint}
\ee
where we assumed that prior to gluing ${\cal T}_2$ consists of labels $\{n_1{+}1, n_1{+}2, \ldots, n_1 {+} n_2\}$. The equality was obtained by multiplying by
\be
1= x_{yz}^2 \int  d\theta_{y} d\theta_{z}  d\chi_{y} d\chi_{z}\, \frac{  \theta_{y} \theta_{z}  \chi_{y} \chi_{z} }{x_{yz}^2}.
\ee
Hence the resulting leading singularity ${\cal T}_3$ from \eqref{gluing-by-edge} has a rational function factorizing as a product of the two rational functions of ${\cal T}_1$ and ${\cal T}_2$:
\be
F_{{\cal T}_3} = x_{yz}^2\, F_{{\cal T}_1} F_{{\cal T}_2} \bigg|_{\substack{v=y\\ w=z}}.
\ee

Let us consider an example in which we can take two copies of the octahedral set of triples:
\begin{align}
{\cal T}_1 &= \{(1,2,3),(3,4,5),(5,6,1),(6,4,2)\},\nonumber\\
{\cal T}_2 &= \{(7,8,9),(9,10,11),(11,12,7),(8,12,10)\}\label{app-example-triples}
\end{align}
and attach them by legs respectively $5\leftrightarrow 12$ and $6 \leftrightarrow 11$. The resulting leading singularity is determined by the data:
\be\label{two-octahedra}
{\cal T}_1 \overset{(5,6)}{\underset{(12,11)}{\circ}} {\cal T}_2 = \{(1, 2, 3), (3, 4, 5), (5, 6, 1), (6, 4, 2), (7, 8, 9), (9, 10, 6), (6, 5, 7), (8, 5, 10)\}.
\ee
The corresponding rational function is given in \eqref{two-octahedra-rtnl}.

\subsection{Gluing by a Triangle}

A natural generalization of the above procedure is to glue two sets of triples by a triangle. Let us identify the triangle $(y,z,t)$ from ${\cal T}_1$ with a triangle $(u,v,s)$ from ${\cal T}_2$ and denote the result with
\be
{\cal T}_3 = {\cal T}_1 \overset{(y,z,t)}{\underset{(v,w,s)}{\circ}} {\cal T}_2.
\ee
In this case, by using the identity
\be
1= x_{yz} x_{zt} x_{ty} \int  d\theta_{y} d\theta_{z} d\theta_t  d\chi_{y} d\chi_{z} d\chi_t\, \Delta_{yzt} \frac{  \theta_{y} \theta_{z}  \chi_{y} \chi_{z} }{x_{yz}^2}.
\ee
it is straightforward to show that the right factor in \eqref{LS-prime-factor} simplifies and the resulting rational function becomes:
\be
F_{{\cal T}_3} = x_{yz} x_{zt} x_{ty}\, F_{{\cal T}_1} F_{{\cal T}_2} \bigg|_{\substack{v=y\\ w=z\\ s=t}}.
\ee

For instance, we can consider the triples from \eqref{app-example-triples} glued by the triangles $(5,6,1) \leftrightarrow (11,12,7)$ yields
\be
{\cal T}_1 \overset{(5,6,1)}{\underset{(11,12,7)}{\circ}} {\cal T}_2 = \{ (1, 2, 3), (3, 4, 5), (5, 6, 1), (6, 4, 2), (1, 8, 9), (9, 10, 5), \
(8, 6, 10)\}.
\ee
The corresponding rational function is
\be
-\frac{\left(x_{15} x_{26} x_{34}-x_{16} x_{24} x_{35}\right)^2 \left(x_{18} x_{59} x_{6,10}-x_{19} x_{5,10}
   x_{68}\right)^2}{x_{12} x_{13} x_{15} x_{16} x_{18} x_{19} x_{23} x_{24} x_{26} x_{34} x_{35} x_{45} x_{46}
   x_{56} x_{59} x_{5,10} x_{68} x_{6,10} x_{89} x_{8,10} x_{9,10}}.
\ee

More generally, one can glue two leading singularities ${\cal T}_1$ and ${\cal T}_2$ by a common leading singularity, if there exists one. It is straightforward to show that the resulting rational function always factors into a product of rational functions $F_{{\cal T}_1} F_{{\cal T}_2}$ times a factor that depends on the gluing object.

\section{\label{app:PT-identities}Derivation of $\mathcal{O}(\epsilon^2)$ Relations}

Let us write the expression in \eqref{genkin} as $I_{n+1}+\Gamma_{n+1}$, where
\begin{equation}
I_{n+1}=A_{12}u_{12}+A_{23}u_{23}+\ldots+A_{n1}u_{n1}\,,
\end{equation}
and we can expand $\Gamma_{n+1}$ in its Grassmann variables as
\begin{equation}
\Gamma_{n+1}=E_{n+1}\theta_{n+1}\chi_{n+1}+\varXi_{n+1}\chi_{n+1}+\Phi_{n+1}\theta_{n+1}+\gamma_{n+1},
\end{equation}
where $\Xi_{n+1}$ and $\Phi_{n+1}$ are 1-forms,
\begin{equation}
\varXi_{n+1} = -\sum_{j=1}^{n}\frac{s_{n+1,j}}{x_{n+1,j}}\theta_{j},\qquad \Phi_{n+1} = -\sum_{j=1}^{n}\frac{s_{n+1,j}}{x_{n+1,j}}\chi_{j}
\end{equation}
and $\gamma_{n+1}$ is a 2-form (whose expression is not needed here).
As already shown $E_{n+1}$ corresponds to the scattering equation
for the $n+1$-th particle after putting $A_{a,a+1}-A_{a-1,a}=\epsilon\,s_{n+1,a}$:
\begin{equation}
  A_{ii+1}=1+\epsilon\sum_{j=1}^{i}s_{n+1,j}  
\end{equation}
(recall we can set $A_{n1}=1$). Note that this requires $\sum_{j=1}^{n}s_{n+1,j}=0$.
Now let us assume $E_{n+1}=\sum_{j=1}^{i}\frac{s_{n+1,j}}{x_{n+1,j}}=0$
holds and consider the power $(I_{n+1}+\Gamma_{n+1})^{n-1}$ up to
second order in $\epsilon$. The RHS (see eq. 2.12) gives again a
linear combination of PT factors:
\begin{align}
& \frac{\theta_{d}\theta_{e}\chi_{d}\chi_{e}}{x_{de}^{2}}(I_{n+1}+\Gamma_{n+1})^{n-1}\\
& = \epsilon^{2}(n{-}1)!\left\lbrack\left(\sum_{i=1}^{n-1}s_{n+1,i}^{2}\right){\rm PT}(1,2,\ldots,n{+}1,n)+\cdots +  s_{n+1,1}^{2}{\rm PT}(1,n{+}1,2,\ldots,n)\right\rbrack+\mathcal{O}(\epsilon^{3}).\nonumber
\end{align}
Here we think of the RHS as Grassmann integrands, i.e., top forms in $\theta$ and $\chi$.

We now concentrate on the binomial expansion of the LHS,
\be
\frac{\theta_{d}\theta_{e}\chi_{d}\chi_{e}}{x_{de}^{2}}(I_{n+1}+\Gamma_{n+1})^{n-1} = \frac{(n-1)(n-2)}{2}\frac{\theta_{d}\theta_{e}\chi_{d}\chi_{e}}{x_{de}^{2}}I_{n+1}^{n-3}\Gamma_{n+1}^{2}+\ldots,
\ee
where the ellipsis represents higher order terms as well as terms independent
of $\theta_{n+1}\chi_{n+1}$. To leading order in $\epsilon$ we have
$I_{n+1}\rightarrow\Delta_{1\ldots n}$ and hence the above expansion
becomes
\begin{align}
\frac{\theta_{d}\theta_{e}\chi_{d}\chi_{e}}{x_{de}^{2}}(I_{n+1}+\Gamma_{n+1})^{n-1}=\epsilon^{2}(n-1)(n-2)&\frac{\theta_{d}\theta_{e}\theta_{n+1}\chi_{d}\chi_{e}\chi_{n+1}}{x_{de}^{2}}(\Delta_{1\ldots n})^{n-3}\nonumber\\
&\qquad\times\sum_{i,j=1}^{n}\frac{s_{n+1,i}\,s_{n+1,j}}{x_{n+1,i}\,x_{n+1,j}}\,\chi_{i}\theta_{j}+\ldots.
\end{align}
Since $\chi_{j}\theta_{j}\sum_{i=1}^{n}\frac{s_{n+1,i}}{x_{n+1,i}}\,=0$
we can write
\begin{align}
\sum_{i,j=1}^{n}\frac{s_{n+1,i}\,s_{n+1,j}}{x_{n+1,i}\,x_{n+1,j}}\,\chi_{i}\theta_{j} & = \sum_{i,j=1}^{n}\frac{s_{n+1,i}\,s_{n+1,j}}{x_{n+1,i}\,x_{n+1,j}}\,(\chi_{i}-\chi_{j})\theta_{j}\nonumber\\
 & = \frac{1}{2}\sum_{i,j=1}^{n}\frac{s_{n+1,i}\,s_{n+1,j}}{x_{n+1,i}\,x_{n+1,j}}\,\chi_{ij}\theta_{ij}\nonumber\\
 & = \frac{1}{2}\sum_{i,j=1}^{n}\frac{x_{ij}}{x_{n+1,i}\,x_{n+1,j}}\,s_{n+1,i}\,s_{n+1,j}\,u_{ij}\\
 & = \frac{1}{2}\sum_{i,j=1}^{n}\frac{x_{ij}}{x_{n+1,i}\,x_{n+1,j}}\,s_{n+1,i}\,s_{n+1,j}\,\Delta_{ij\,n+1}\,,\nonumber
\end{align}
where in the last line we added the term $u_{j,n+1}+u_{n+1,i}$ which
vanishes inside the sum. For instance:
\begin{align}
\frac{1}{2}\sum_{i,j=1}^{n}\frac{x_{ij}}{x_{n+1,i}\,x_{n+1,j}}\,s_{n+1,i}\,s_{n+1,j}\,u_{j,n+1}=&\frac{1}{2}\sum_{i,j=1}^{n}\frac{s_{n+1,i}s_{n+1,j}}{x_{n+1,j}}u_{j,n+1}\nonumber\\
&-\frac{1}{2}\sum_{i,j=1}^{n}\frac{s_{n+1,i}s_{n+1,j}}{x_{n+1,i}}u_{j,n+1}=0\,,
\end{align}
thanks to both momentum conservation and the scattering equations.
We then have
\begin{align}
&\frac{\theta_{d}\theta_{e}\chi_{d}\chi_{e}}{x_{de}^{2}}(I_{n+1}+\Gamma_{n+1})_{\epsilon^{2}}^{n-1}\nonumber\\
& = \epsilon^{2}\frac{(n-1)(n-2)}{2}\frac{\theta_{d}\theta_{e}\theta_{n+1}\chi_{d}\chi_{e}\chi_{n+1}}{x_{de}^{2}}(\Delta_{1\ldots n})^{n-3}\sum_{i,j=1}^{n}\frac{x_{ij}}{x_{n+1,i}\,x_{n+1,j}}\,s_{n+1,i}\,s_{n+1,j}\,\Delta_{i,j,n+1}\nonumber\\
 & = \epsilon^{2}(n-1)(n-2)\frac{\theta_{d}\theta_{e}\chi_{d}\chi_{e}}{x_{de}^{2}}\Lambda_{de}\Delta_{d,e,n+1}(\Delta_{1\ldots n})^{n-3}\sum_{i<j}^{n}\Lambda_{ij}^{-1}\,s_{n+1,i}\,s_{n+1,j}\,\Delta_{i,j,n+1},
\end{align}
where $\Lambda_{ij}=x_{n+1,i}x_{n+1,j}/x_{i,j}$ is the inverse
soft factor, and we used the support of $\theta_{d}\theta_{e}\chi_{d}\chi_{e}$.
Let us now triangulate $\Delta_{n+1,d,e}$ with $n$ triples:
\be
\Delta_{1\ldots n}=\sum_{i}\Delta_{i,i+1,n+1}
\ee
so that
\be
(\Delta_{1\ldots n})^{n-3}=(n-3)!\sum_{p<q<r}\prod_{l}\Delta_{l,l+1,n+1}\hat{\Delta}_{p,p+1,n+1}\hat{\Delta}_{q,q+1,n+1}\hat{\Delta}_{r,r+1,n+1},
\ee
i.e., we remove the $p^\text{th}$, $q^\text{th}$ and $r^\text{th}$ triple. Now fix the
labels $\{d,e\}=\{k,k+1\}$ for some $k$, i.e., we have an overall
factor of $\Lambda_{k,k+1}\Delta_{k,k+1,n+1}$ (the prefactor $\theta_{d}\theta_{e}\chi_{d}\chi_{e}/x_{de}^{2}$
is just a Jacobian and can be ignored once we write the top form in
terms of triples). If the set of triangles inside the product contains
$\Delta_{k,k+1,n+1}$ then such contribution vanishes, hence we can
restrict the sum to the cases $k\in\{p,q,r\}$. This can be written
as
\be
\epsilon^{2}(n-1)!\,\Lambda_{k,k+1}\sum_{\substack{p<q\\
p,q\neq k
}
}\Lambda_{ij}^{-1}\,s_{n+1,i}\,s_{n+1,j}\,\Delta_{i,j,n+1}\prod_{l}\Delta_{l,l+1,n+1}\hat{\Delta}_{p,p+1,n+1}\hat{\Delta}_{q,q+1,n+1}
\ee
Merging the sums we find
\begin{align}
\epsilon^{2}(n-1)!\,\Lambda_{n,1}\left(\sum_{i\leq p<j\leq q<n
}+\sum_{
p<j\leq q<i\leq n
}\right)\Lambda_{ij}^{-1}&\,s_{n+1,i}\,s_{n+1,j}\,\Delta_{i,j,n+1} \nonumber\\
&\times\prod_{l}\Delta_{l,l+1,n+1}\hat{\Delta}_{p,p+1,n+1}\hat{\Delta}_{q,q+1,n+1}.
\end{align}

Here we chose $k=n$ for simplicity. It can be seen that all the other
summation regions vanish by the fact that in those cases we can triangulate
the set of triangles inside the product such that they contain
$\Delta_{i,j,n+1}$. We now recognize in each term in the sum the
product of three leading singularities which triples read
\begin{align}
{\cal T}_{1} & = \{(a,b,c)\},\nonumber\\
{\cal T}_{2} & = \{(n+1,p+1,p+2),\ldots,(n+1,j,j+1),\ldots,(n+1,q-1,q)\},\\
{\cal T}_{3} & = \{(d,q+1,q+2),\ldots,(d,i,i+1),\ldots,(d,p-1,p)\}.\nonumber
\end{align}
Using the factorization property find the leading singularity function to be
\begin{align}
F\left(\left[{\cal T}_{1}\overset{(a,c)}{\underset{(i,n+1)}{\circ}}{\cal T}_{3}\right]\overset{(n+1,b)}{\underset{(d,j)}{\circ}}{\cal T}_{2}\right) & = x_{n+1,j}^{2}\mbox{{\rm PT}}(q+1,\ldots p,n+1)\, F \left[{\cal T}_{1}\overset{(a,c)}{\underset{(i,n+1)}{\circ}}{\cal T}_{3}\right]_{b=j}\\
& = \underbrace{\frac{x_{n+1,j}x_{n+1,i}}{x_{i,j}}}_{\Lambda_{ij}}\mbox{{\rm PT}}(q+1,\ldots p,n+1)\mbox{{\rm PT}}(p+1,\ldots q,n+1),\nonumber
\end{align}
where we used the gluing operation introduced in Appendix~\ref{app:gluing}. Inserting this back into the sum we arrive at
\begin{align}
&\frac{\theta_{d}\theta_{e}\chi_{d}\chi_{e}}{x_{de}^{2}}(I_{n+1}+\Gamma_{n+1})^{n-1} \Big|_{{\epsilon^{2}}}\nonumber \\
& =  \epsilon^{2}(n{-}1)!\Lambda_{n,1}\!\!\left(\!\sum_{i\leq p<j\leq q<n} \!\!+\!\! \sum_{p<j\leq q<i\leq n}\!\right)\!\! s_{n+1,i}s_{n+1,j}\mbox{{\rm PT}}(q{+}1,\ldots, p,n{+}1)\mbox{{\rm PT}}(p{+}1,\ldots, q,n{+}1)\nonumber\\
& = \epsilon^{2}(n{-}1)!\Lambda_{n,1}\sum_{1\leq p<q<n}\,s_{n+1,(p+1\ldots q)}^{2}\,\mbox{{\rm PT}}(q+1,\ldots, p,n+1)\mbox{{\rm PT}}(p+1,\ldots, q,n+1),
\end{align}
where $s_{n+1,(p+1\ldots q)}=s_{n+1,(q+1\ldots p)}:=s_{n+1,p+1}+s_{n+1,p+2}+\ldots+s_{n+1,q}$.
Combining both expansions we finally get
\begin{align}
\Lambda_{n,1}\sum_{1\leq p<q<n}&\,s_{n+1,(p+1\ldots q)}^{2}\, \mbox{{\rm PT}}(q+1,\ldots p,n+1)\mbox{{\rm PT}}(p+1,\ldots q,n+1)\nonumber\\
= & \left(\sum_{i=1}^{n-1}s_{n+1,i}^{2}\right){\rm PT}(1,2,\ldots,n+1,n)+\ldots +s_{n+1,1}^{2}{\rm PT}(1,n+1,2,\ldots,n)\,.
\end{align}

\section{\label{app:shuffle}Proof of Shuffle Identities}

In this appendix we give a new combinatorial proof of the (physical) Kleiss--Kuijf relations using the $\Delta$-algebra, together with the homomorphism defined by
\be
\text{LS}_{\cal T} \mapsto \text{LS}_{\cal T}\, \frac{\theta_a \theta_b}{x_{ab}} \frac{\chi_c \chi_d}{x_{cd}}.
\ee
This is in fact an isomorphism of symmetric group representations which is independent of the choice of pairs of distinct elements $(a,b)$ and $(c,d)$; we thereby obtain a new proof of the physical KK relations.

For any permutation $\omega$ of the set $\{2,3,\ldots, n{-}1\}$, recall the definition
\begin{equation}
{d_{\omega}}=u_{1,\omega^{(1)}}+u_{\omega^{(1)},\omega^{(2)}}+\ldots+u_{\omega^{(n-2)},n}+u_{n,1}\,
\end{equation}
and further set
\begin{align}
d :=\,&  u_{12}+\ldots+u_{m,n}+u_{n,m+1}+\ldots+u_{n-2,n-1}+u_{n-1,1}\nonumber \\
 =& \Delta_{123}+ \ldots+\Delta_{1mn}+\Delta_{1,n,m+1}+\ldots+\Delta_{1,n-2,n-1}\,,
\end{align}
corresponding to the permutation 
\be
(1,2,\ldots, m,n,m+1,m+2,\ldots, n-2,n-1).
\ee
Theorem \ref{thm:Delta KK} gives the expansion of $\frac{1}{(n-2)!}d^{n-2}$ in the KK basis,
\be
\left\{\frac{1}{(n-2)!}d_\omega^{n-2} \;\Big|\; \omega \text{ is a permutation of }(2,\ldots, n-1)\right\}.
\ee

\begin{theorem}\label{thm:Delta KK}
We have
\begin{align}\label{eq:KK relation Delta algebra}
\sum_{i=1}^{\binom{n-2}{m-1}}d_{\omega_{i}}^{n-2} = (-1)^{n-m-1}d^{n-2}\,,
\end{align}
where $\omega_i$ varies over the set of shuffles $\alpha \shuffle \beta^T$ of $\alpha=(2,3,\ldots, m)$ and $\beta^T = (n{-}1,n{-}2,\ldots,\allowbreak m{+}1)$.
\end{theorem}

\begin{proof}
The first step is to construct the triangulations of $\boxempty_{\omega_i}$
and $\hat{\boxempty}_{\omega_i}$. For this we note that each $\boxempty_{\omega_i}$
represents a path from label $1$ to label $n$ in the ``two-route''
arrangement of Figure~\ref{fig:random-walk}. Each path passes through all the labels,
and covers the labels in each route in a consecutive order. 

The shaded region we call the ``interior'' (exterior) of the path,
and we associate the complex $\boxempty_{\omega_i}$ ($\hat{\boxempty}_{\omega_i}$). It is easy to check that constructed in this way $\boxempty_{\omega_i}$ consists of $m-1$ triangles. Let us denote by $\{\Delta_{j}^{(i)}\}_{j=1}^{m-1}$
some canonical triangulation of the interior of the $i^{\text{th}}$ path,
i.e., $\boxempty_{\omega_i}=\sum_{j}\Delta_{j}^{(i)}$. Then we can express
the KK sum as a ``sum over paths'':
\be
\sum_{i=1}^{\binom{n-2}{m-1}}\boxempty_{\omega_i}^{m-1}=(m-1)!\sum_{i=1}^{\binom{n-2}{m-1}}\prod_{j=1}^{m-1}\Delta_{j}^{(i)}
\ee
\begin{figure}
\centering
\includegraphics[scale=0.9]{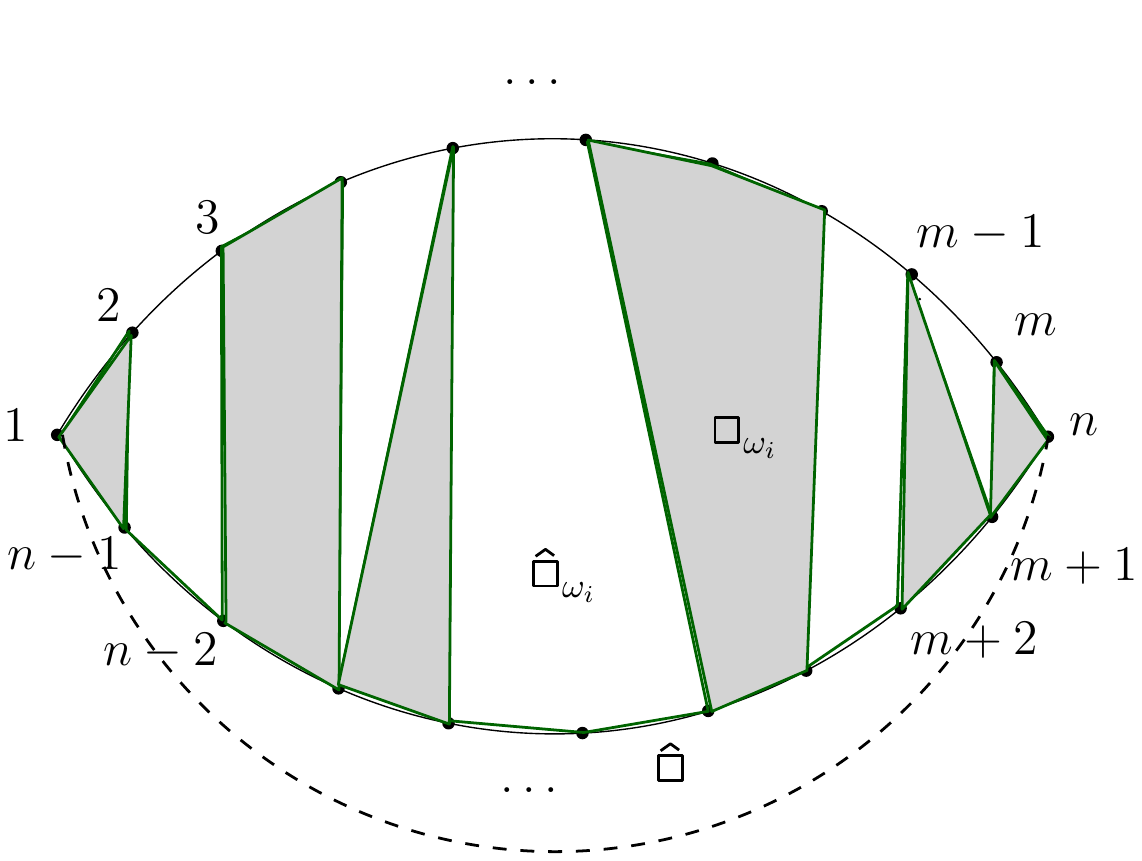}
\caption{\label{fig:random-walk} A ``random walk'' from $1$ to $n$. Recall that $\boxempty_{\omega_i}+\hat{\boxempty}_{\omega_i}=\boxempty+\hat{\boxempty}$
and $d_{\omega_i}=\boxempty_{\omega_i}-\hat{\boxempty}$.}
\end{figure}

We are now in position to prove our main identities. Define the function:
\be
f_{n,m}(\{\Delta_{j}^{(i)}\}):=\sum_{i=1}^{\binom{n-2}{m-1}}\boxempty_{\omega_{i}}^{m-1}-d_n^{m-1},
\ee
where $d_n=\boxempty+\hat{\boxempty}$. Note that such function
can be expanded explicitly in terms of, in principle, all the possible
triangles $\{\Delta_{j}\}$ with a side laying on one of the routes,
since we can always triangulate $\boxempty_{\omega_i}$ using such. Note
also that $f_{n,m}(\{\Delta_{j}^{(i)}\})$ is of uniform Grassmann
degree, and thanks to the nilpotency of $\Delta$'s it is at most
linear in each $\Delta_{j}$. Hence, it is sufficient to show that
the coefficient of each $\Delta_{j}$ vanishes to argue that $f_{n,m}(\{\Delta_{j}^{(i)}\})=0$.

We will use an inductive argument in $n$ and $m$. The cases $m=2$
are trivial to check for any $n$ since they correspond to $U(1)$
decoupling. Let us then assume $f_{n-1,m-1}=0$ for any set of triangles.
Consider then a given $\Delta^{(*)}$ that is contained among a subset
of paths $P_{\Delta^{(*)}}=\{\boxempty_{i}\}_{i\in I}$. We can single
out the contribution from $\text{\ensuremath{\Delta}}^{(*)}$ in the
sum over paths as
\be
\sum_{i=1}^{\binom{n-2}{m-1}}\boxempty_{\omega_{i}}^{m-1}=(m-1)!\Delta^{(*)}\sum_{i\in I}\prod_{j=1}^{m-2}\Delta_{j}^{(i)}+\ldots.
\ee

Now, it is easy to check that removing $\Delta^{(*)}$ from a path
$i\in I$ defines a new path for $n-1$ labels, whose interior contains
$m-2$ triangles (see Figure~\ref{fig:singling}). In fact we can now see that the number
of such paths is $|I|=\binom{n-3}{m-2}$. Since $f_{n-1,m-1}=0$ we
can write 
\be
\sum_{i\in I}\prod_{j=1}^{m-2}\Delta_{j}^{(i)}=\frac{1}{(m-2)!}d_{n-1}^{m-2}.
\ee
\begin{figure}
\centering
\includegraphics[scale=0.5]{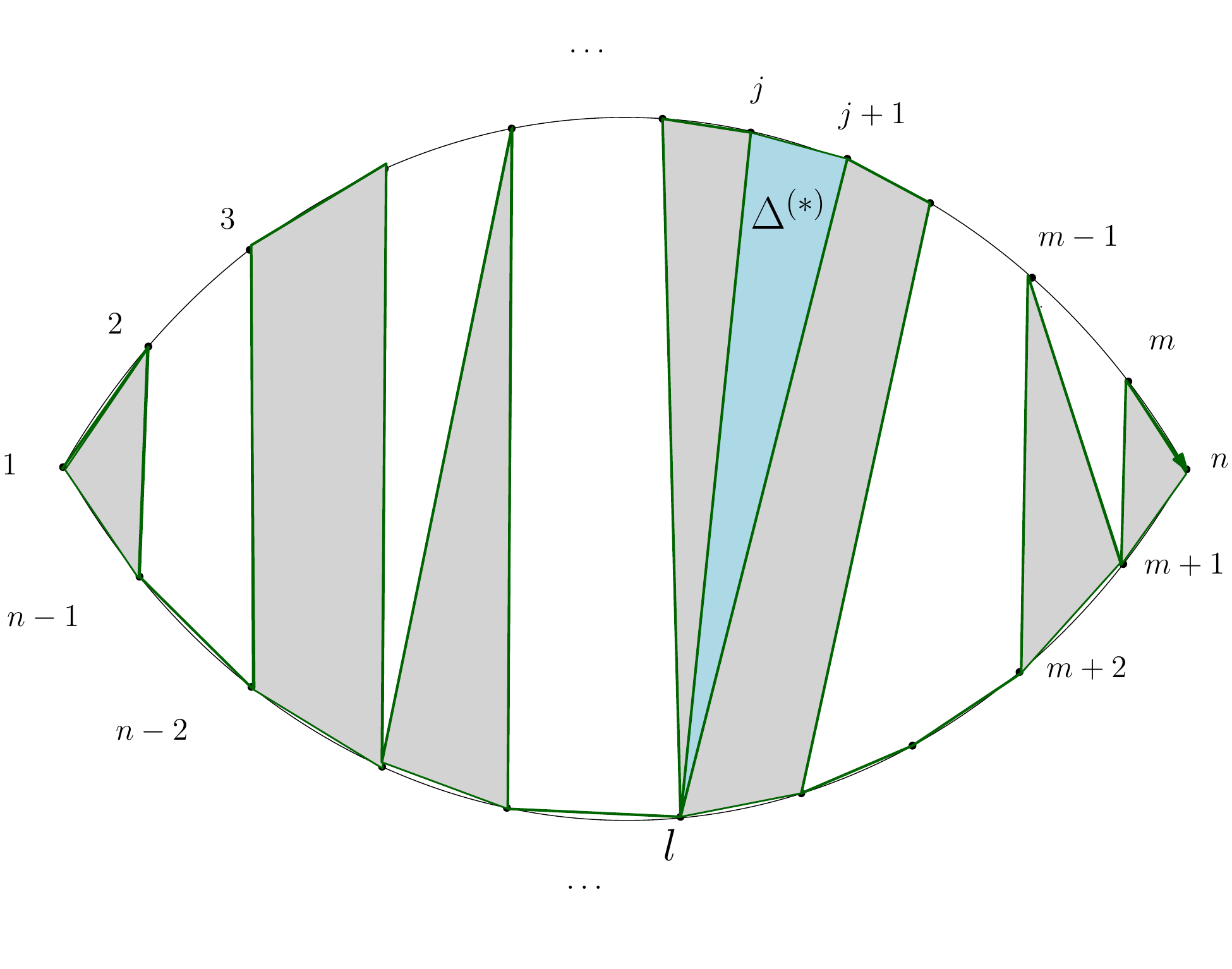}
\caption{\label{fig:singling}Singling out $\Delta^{(*)}$. The path obtained
by identifying label $j$ with $j+1$ contains $m-2$ triangles and
$n-1$ labels.}
\end{figure}
where $d_{n-1}$ is now the complement region of $\Delta^{(*)}$.
We then have 
\be
\sum_{i=1}^{\binom{n-2}{m-1}}\boxempty_{\omega_{i}}^{m-1}=(m-1)\Delta^{(*)}d_{n-1}^{m-2}+\ldots.
\ee
Now write $d_{n}=\Delta^{(*)}+d_{n-1}$ such that we have
\be
d_{n}^{m-1}=d_{n-1}^{m-1}+(m-1)\Delta^{(*)}d_{n-1}^{m-2}.
\ee

We can now see that the coefficient of the second term precisely cancels
the contribution from the path sum, hence $f_{n,m}=0$. This concludes
the proof.
\end{proof}

As a corollary of our construction, we note that, given a triangulation of the region $d_n$, each of the triangles will appear $|I|=\binom{n-3}{m-2}$ times when summing over all possible paths, that is
\be
\sum_{i=1}^{\binom{n-2}{m-1}}\boxempty_{\omega_{i}}=\binom{n-3}{m-2} d_n\,.
\ee
This proves our formula \eqref{eq:linearbox}.

It is interesting to note that the idea of the proof can be extended
to a general piece (i.e., subset of triangles) of $\Omega_{n}$, say
$\Omega_{n}=R+\hat{R}$. If $R$ has $k$ triangles then $\Omega_{n}^{m-1}=\binom{m-1}{k}R^{k}\hat{R}^{m-1-k}+\ldots{\rm (lower\,\,orders)}.$
The binomial in the LHS of the identity is unique in that it accounts for all the combinatorial factors associated to singling out $R^{k}$ in the path sum.

\bibliographystyle{JHEP}
\bibliography{references}

\providecommand{\href}[2]{#2}\begingroup\raggedright\begin{thebibliography}{10}

\bibitem{ArkaniHamed:2012nw}
N.~Arkani-Hamed, J.~L. Bourjaily, F.~Cachazo, A.~B. Goncharov, A.~Postnikov and
  J.~Trnka, \emph{{Grassmannian Geometry of Scattering Amplitudes}}. Cambridge
  University Press, 2016,
  \href{https://doi.org/10.1017/CBO9781316091548}{10.1017/CBO9781316091548},
  [\href{https://arxiv.org/abs/1212.5605}{{\ttfamily 1212.5605}}].

\bibitem{postnikov2006total}
A.~Postnikov, \emph{{Total positivity, Grassmannians, and networks}},
  \href{https://arxiv.org/abs/math/0609764}{{\ttfamily math/0609764}}.

\bibitem{Arkani-Hamed:2014bca}
N.~Arkani-Hamed, J.~L. Bourjaily, F.~Cachazo, A.~Postnikov and J.~Trnka,
  \emph{{On-Shell Structures of MHV Amplitudes Beyond the Planar Limit}},
  \href{https://doi.org/10.1007/JHEP06(2015)179}{\emph{JHEP} {\bfseries 06}
  (2015) 179} [\href{https://arxiv.org/abs/1412.8475}{{\ttfamily 1412.8475}}].

\bibitem{He:2018okq}
S.~He and C.~Zhang, \emph{{Notes on Scattering Amplitudes as Differential
  Forms}}, \href{https://doi.org/10.1007/JHEP10(2018)054}{\emph{JHEP}
  {\bfseries 10} (2018) 054}
  [\href{https://arxiv.org/abs/1807.11051}{{\ttfamily 1807.11051}}].

\bibitem{Britto:2004nc}
R.~Britto, F.~Cachazo and B.~Feng, \emph{{Generalized unitarity and one-loop
  amplitudes in N=4 super-Yang-Mills}},
  \href{https://doi.org/10.1016/j.nuclphysb.2005.07.014}{\emph{Nucl. Phys.}
  {\bfseries B725} (2005) 275}
  [\href{https://arxiv.org/abs/hep-th/0412103}{{\ttfamily hep-th/0412103}}].

\bibitem{Buchbinder:2005wp}
E.~I. Buchbinder and F.~Cachazo, \emph{{Two-loop amplitudes of gluons and
  octa-cuts in N=4 super Yang-Mills}},
  \href{https://doi.org/10.1088/1126-6708/2005/11/036}{\emph{JHEP} {\bfseries
  11} (2005) 036} [\href{https://arxiv.org/abs/hep-th/0506126}{{\ttfamily
  hep-th/0506126}}].

\bibitem{Cachazo:2008vp}
F.~Cachazo, \emph{{Sharpening The Leading Singularity}},
  \href{https://arxiv.org/abs/0803.1988}{{\ttfamily 0803.1988}}.

\bibitem{Arkani-Hamed:2017tmz}
N.~Arkani-Hamed, Y.~Bai and T.~Lam, \emph{{Positive Geometries and Canonical
  Forms}}, \href{https://doi.org/10.1007/JHEP11(2017)039}{\emph{JHEP}
  {\bfseries 11} (2017) 039}
  [\href{https://arxiv.org/abs/1703.04541}{{\ttfamily 1703.04541}}].

\bibitem{Arkani-Hamed:2013jha}
N.~Arkani-Hamed and J.~Trnka, \emph{{The Amplituhedron}},
  \href{https://doi.org/10.1007/JHEP10(2014)030}{\emph{JHEP} {\bfseries 10}
  (2014) 030} [\href{https://arxiv.org/abs/1312.2007}{{\ttfamily 1312.2007}}].

\bibitem{Arkani-Hamed:2017vfh}
N.~Arkani-Hamed, H.~Thomas and J.~Trnka, \emph{{Unwinding the Amplituhedron in
  Binary}}, \href{https://doi.org/10.1007/JHEP01(2018)016}{\emph{JHEP}
  {\bfseries 01} (2018) 016}
  [\href{https://arxiv.org/abs/1704.05069}{{\ttfamily 1704.05069}}].

\bibitem{Ferro:2018vpf}
L.~Ferro, T.~Lukowski and M.~Parisi, \emph{{Amplituhedron meets Jeffrey-Kirwan
  Residue}},  \href{https://arxiv.org/abs/1805.01301}{{\ttfamily 1805.01301}}.

\bibitem{early2018configuration}
N.~Early and V.~Reiner, \emph{{On configuration spaces and Whitehouse's lifts
  of the Eulerian representations}},
  \href{https://arxiv.org/abs/1808.04007}{{\ttfamily 1808.04007}}.

\bibitem{moseley2017orlik}
D.~Moseley, N.~Proudfoot and B.~Young, \emph{{The Orlik-Terao algebra and the
  cohomology of configuration space}}, {\emph{Experimental Mathematics}
  {\bfseries 26} (2017) 373}
  [\href{https://arxiv.org/abs/1603.01189}{{\ttfamily 1603.01189}}].

\bibitem{knudsen2018configuration}
B.~Knudsen, \emph{Configuration spaces in algebraic topology},
  \href{https://arxiv.org/abs/1803.11165}{{\ttfamily 1803.11165}}.

\bibitem{Arnol'd1969}
V.~I. Arnol'd, \emph{The cohomology ring of the colored braid group},
  \href{https://doi.org/10.1007/BF01098313}{\emph{Mathematical notes of the
  Academy of Sciences of the USSR} {\bfseries 5} (1969) 138}.

\bibitem{10.2307/2946631}
W.~Fulton and R.~MacPherson, \emph{A compactification of configuration spaces},
  {\emph{Annals of Mathematics} {\bfseries 139} (1994) 183}.

\bibitem{TOTARO19961057}
B.~Totaro, \emph{Configuration spaces of algebraic varieties},
  \href{https://doi.org/https://doi.org/10.1016/0040-9383(95)00058-5}{\emph{Topology}
  {\bfseries 35} (1996) 1057 }.

\bibitem{kriz1994rational}
I.~Kriz, \emph{On the rational homotopy type of configuration spaces},
  {\emph{Annals of Mathematics} {\bfseries 139} (1994) 227}.

\bibitem{OcneanuLectures}
A.~Ocneanu, ``{Higher Representation Theory in Math and Physics}.'' Harvard
  University course PHYSICS 267, Fall 2017,
  \url{https://youtu.be/9gHzFLfPFFU?t=380}.

\bibitem{Early:2018mac}
N.~Early, \emph{{Honeycomb tessellations and canonical bases for permutohedral
  blades}},  \href{https://arxiv.org/abs/1810.03246}{{\ttfamily 1810.03246}}.

\bibitem{Early:2017lku}
N.~Early, \emph{{Generalized Permutohedra, Scattering Amplitudes, and a Cubic
  Three-Fold}},  \href{https://arxiv.org/abs/1709.03686}{{\ttfamily
  1709.03686}}.

\bibitem{Enciso:2014cta}
M.~Enciso, \emph{{Volumes of Polytopes Without Triangulations}},
  \href{https://doi.org/10.1007/JHEP10(2017)071}{\emph{JHEP} {\bfseries 10}
  (2017) 071} [\href{https://arxiv.org/abs/1408.0932}{{\ttfamily 1408.0932}}].

\bibitem{Enciso:2016cif}
M.~Enciso, \emph{{Logarithms and Volumes of Polytopes}},
  \href{https://doi.org/10.1007/JHEP04(2018)016}{\emph{JHEP} {\bfseries 04}
  (2018) 016} [\href{https://arxiv.org/abs/1612.07370}{{\ttfamily
  1612.07370}}].

\bibitem{Hodges:2009hk}
A.~Hodges, \emph{{Eliminating spurious poles from gauge-theoretic amplitudes}},
  \href{https://doi.org/10.1007/JHEP05(2013)135}{\emph{JHEP} {\bfseries 05}
  (2013) 135} [\href{https://arxiv.org/abs/0905.1473}{{\ttfamily 0905.1473}}].

\bibitem{Feng:2011np}
B.~Feng and M.~Luo, \emph{{An Introduction to On-shell Recursion Relations}},
  \href{https://doi.org/10.1007/s11467-012-0270-z}{\emph{Front. Phys.(Beijing)}
  {\bfseries 7} (2012) 533} [\href{https://arxiv.org/abs/1111.5759}{{\ttfamily
  1111.5759}}].

\bibitem{Elvang:2013cua}
H.~Elvang and Y.-t. Huang, \emph{{Scattering Amplitudes}},
  \href{https://arxiv.org/abs/1308.1697}{{\ttfamily 1308.1697}}.

\bibitem{Bern:2008qj}
Z.~Bern, J.~J.~M. Carrasco and H.~Johansson, \emph{{New Relations for
  Gauge-Theory Amplitudes}},
  \href{https://doi.org/10.1103/PhysRevD.78.085011}{\emph{Phys. Rev.}
  {\bfseries D78} (2008) 085011}
  [\href{https://arxiv.org/abs/0805.3993}{{\ttfamily 0805.3993}}].

\bibitem{BjerrumBohr:2009rd}
N.~E.~J. Bjerrum-Bohr, P.~H. Damgaard and P.~Vanhove, \emph{{Minimal Basis for
  Gauge Theory Amplitudes}},
  \href{https://doi.org/10.1103/PhysRevLett.103.161602}{\emph{Phys. Rev. Lett.}
  {\bfseries 103} (2009) 161602}
  [\href{https://arxiv.org/abs/0907.1425}{{\ttfamily 0907.1425}}].

\bibitem{BjerrumBohr:2010zs}
N.~E.~J. Bjerrum-Bohr, P.~H. Damgaard, T.~Sondergaard and P.~Vanhove,
  \emph{{Monodromy and Jacobi-like Relations for Color-Ordered Amplitudes}},
  \href{https://doi.org/10.1007/JHEP06(2010)003}{\emph{JHEP} {\bfseries 06}
  (2010) 003} [\href{https://arxiv.org/abs/1003.2403}{{\ttfamily 1003.2403}}].

\bibitem{Franco:2015rma}
S.~Franco, D.~Galloni, B.~Penante and C.~Wen, \emph{{Non-Planar On-Shell
  Diagrams}}, \href{https://doi.org/10.1007/JHEP06(2015)199}{\emph{JHEP}
  {\bfseries 06} (2015) 199}
  [\href{https://arxiv.org/abs/1502.02034}{{\ttfamily 1502.02034}}].

\bibitem{ArkaniHamed:2009dn}
N.~Arkani-Hamed, F.~Cachazo, C.~Cheung and J.~Kaplan, \emph{{A Duality For The
  S Matrix}}, \href{https://doi.org/10.1007/JHEP03(2010)020}{\emph{JHEP}
  {\bfseries 03} (2010) 020} [\href{https://arxiv.org/abs/0907.5418}{{\ttfamily
  0907.5418}}].

\bibitem{deligne1999quantum}
P.~Deligne and J.~W. Morgan, \emph{{Notes on Supersymmetry (following Joseph
  Bernstein)}}, Quantum Fields and Strings: A Course for Mathematicians.
  American Mathematical Society, 1999.

\bibitem{Cachazo:2013hca}
F.~Cachazo, S.~He and E.~Y. Yuan, \emph{{Scattering of Massless Particles in
  Arbitrary Dimensions}},
  \href{https://doi.org/10.1103/PhysRevLett.113.171601}{\emph{Phys. Rev. Lett.}
  {\bfseries 113} (2014) 171601}
  [\href{https://arxiv.org/abs/1307.2199}{{\ttfamily 1307.2199}}].

\bibitem{Witten:2003nn}
E.~Witten, \emph{{Perturbative gauge theory as a string theory in twistor
  space}}, \href{https://doi.org/10.1007/s00220-004-1187-3}{\emph{Commun. Math.
  Phys.} {\bfseries 252} (2004) 189}
  [\href{https://arxiv.org/abs/hep-th/0312171}{{\ttfamily hep-th/0312171}}].

\bibitem{Roiban:2004yf}
R.~Roiban, M.~Spradlin and A.~Volovich, \emph{{On the tree level S matrix of
  Yang-Mills theory}},
  \href{https://doi.org/10.1103/PhysRevD.70.026009}{\emph{Phys. Rev.}
  {\bfseries D70} (2004) 026009}
  [\href{https://arxiv.org/abs/hep-th/0403190}{{\ttfamily hep-th/0403190}}].

\bibitem{Parke:1986gb}
S.~J. Parke and T.~R. Taylor, \emph{{An Amplitude for $n$ Gluon Scattering}},
  \href{https://doi.org/10.1103/PhysRevLett.56.2459}{\emph{Phys. Rev. Lett.}
  {\bfseries 56} (1986) 2459}.

\bibitem{Kleiss:1988ne}
R.~Kleiss and H.~Kuijf, \emph{{Multi - Gluon Cross-sections and Five Jet
  Production at Hadron Colliders}},
  \href{https://doi.org/10.1016/0550-3213(89)90574-9}{\emph{Nucl. Phys.}
  {\bfseries B312} (1989) 616}.

\bibitem{Dixon:1996wi}
L.~J. Dixon, \emph{{Calculating scattering amplitudes efficiently}},  in
  \emph{{QCD and beyond. Proceedings, Theoretical Advanced Study Institute in
  Elementary Particle Physics, TASI-95, Boulder, USA, June 4-30, 1995}},
  pp.~539--584, 1996, \href{https://arxiv.org/abs/hep-ph/9601359}{{\ttfamily
  hep-ph/9601359}}.

\bibitem{Cachazo:2012uq}
F.~Cachazo, \emph{{Fundamental BCJ Relation in N=4 SYM From The Connected
  Formulation}},  \href{https://arxiv.org/abs/1206.5970}{{\ttfamily
  1206.5970}}.

\bibitem{orlik2013arrangements}
P.~Orlik and H.~Terao, \emph{Arrangements of Hyperplanes}, Grundlehren der
  mathematischen Wissenschaften. Springer Berlin Heidelberg, 2013.

\bibitem{esnault1992cohomology}
H.~Esnault, V.~Schechtman and E.~Viehweg, \emph{Cohomology of local systems on
  the complement of hyperplanes}, {\emph{Inventiones mathematicae} {\bfseries
  109} (1992) 557}.

\bibitem{Mizera:2017rqa}
S.~Mizera, \emph{{Scattering Amplitudes from Intersection Theory}},
  \href{https://doi.org/10.1103/PhysRevLett.120.141602}{\emph{Phys. Rev. Lett.}
  {\bfseries 120} (2018) 141602}
  [\href{https://arxiv.org/abs/1711.00469}{{\ttfamily 1711.00469}}].

\bibitem{He:2016iqi}
S.~He and Y.~Zhang, \emph{{New Formulas for Amplitudes from Higher-Dimensional
  Operators}}, \href{https://doi.org/10.1007/JHEP02(2017)019}{\emph{JHEP}
  {\bfseries 02} (2017) 019}
  [\href{https://arxiv.org/abs/1608.08448}{{\ttfamily 1608.08448}}].

\bibitem{He:2018pue}
S.~He, G.~Yan, C.~Zhang and Y.~Zhang, \emph{{Scattering Forms, Worldsheet Forms
  and Amplitudes from Subspaces}},
  \href{https://doi.org/10.1007/JHEP08(2018)040}{\emph{JHEP} {\bfseries 08}
  (2018) 040} [\href{https://arxiv.org/abs/1803.11302}{{\ttfamily
  1803.11302}}].

\bibitem{PACHNER1991129}
U.~Pachner, \emph{P.l. homeomorphic manifolds are equivalent by elementary
  shellings},
  \href{https://doi.org/https://doi.org/10.1016/S0195-6698(13)80080-7}{\emph{European
  Journal of Combinatorics} {\bfseries 12} (1991) 129 }.

\bibitem{Arkani-Hamed:2017mur}
N.~Arkani-Hamed, Y.~Bai, S.~He and G.~Yan, \emph{{Scattering Forms and the
  Positive Geometry of Kinematics, Color and the Worldsheet}},
  \href{https://doi.org/10.1007/JHEP05(2018)096}{\emph{JHEP} {\bfseries 05}
  (2018) 096} [\href{https://arxiv.org/abs/1711.09102}{{\ttfamily
  1711.09102}}].

\bibitem{He:2018svj}
S.~He and Q.~Yang, \emph{{An Etude on Recursion Relations and Triangulations}},
   \href{https://arxiv.org/abs/1810.08508}{{\ttfamily 1810.08508}}.

\bibitem{Benincasa:2015zna}
P.~Benincasa, \emph{{On-shell diagrammatics and the perturbative structure of
  planar gauge theories}},  \href{https://arxiv.org/abs/1510.03642}{{\ttfamily
  1510.03642}}.

\bibitem{Heslop:2016plj}
P.~Heslop and A.~E. Lipstein, \emph{{On-shell diagrams for $ \mathcal{N} $ = 8
  supergravity amplitudes}},
  \href{https://doi.org/10.1007/JHEP06(2016)069}{\emph{JHEP} {\bfseries 06}
  (2016) 069} [\href{https://arxiv.org/abs/1604.03046}{{\ttfamily
  1604.03046}}].

\bibitem{Herrmann:2016qea}
E.~Herrmann and J.~Trnka, \emph{{Gravity On-shell Diagrams}},
  \href{https://doi.org/10.1007/JHEP11(2016)136}{\emph{JHEP} {\bfseries 11}
  (2016) 136} [\href{https://arxiv.org/abs/1604.03479}{{\ttfamily
  1604.03479}}].

\bibitem{Arkani-Hamed:2014dca}
N.~Arkani-Hamed, A.~Hodges and J.~Trnka, \emph{{Positive Amplitudes In The
  Amplituhedron}}, \href{https://doi.org/10.1007/JHEP08(2015)030}{\emph{JHEP}
  {\bfseries 08} (2015) 030} [\href{https://arxiv.org/abs/1412.8478}{{\ttfamily
  1412.8478}}].

\bibitem{Britto:2004ap}
R.~Britto, F.~Cachazo and B.~Feng, \emph{{New recursion relations for tree
  amplitudes of gluons}},
  \href{https://doi.org/10.1016/j.nuclphysb.2005.02.030}{\emph{Nucl. Phys.}
  {\bfseries B715} (2005) 499}
  [\href{https://arxiv.org/abs/hep-th/0412308}{{\ttfamily hep-th/0412308}}].

\bibitem{Britto:2005fq}
R.~Britto, F.~Cachazo, B.~Feng and E.~Witten, \emph{{Direct proof of tree-level
  recursion relation in Yang-Mills theory}},
  \href{https://doi.org/10.1103/PhysRevLett.94.181602}{\emph{Phys. Rev. Lett.}
  {\bfseries 94} (2005) 181602}
  [\href{https://arxiv.org/abs/hep-th/0501052}{{\ttfamily hep-th/0501052}}].

\bibitem{reiner2005lectures}
V.~Reiner, ``Lectures on matroids and oriented matroids.''
  \url{http://www-users.math.umn.edu/~reiner/Talks/Vienna05/Lectures.pdf}.

\bibitem{Cachazo:2012da}
F.~Cachazo and Y.~Geyer, \emph{{A 'Twistor String' Inspired Formula For
  Tree-Level Scattering Amplitudes in N=8 SUGRA}},
  \href{https://arxiv.org/abs/1206.6511}{{\ttfamily 1206.6511}}.

\bibitem{Cachazo:2013gna}
F.~Cachazo, S.~He and E.~Y. Yuan, \emph{{Scattering equations and
  Kawai-Lewellen-Tye orthogonality}},
  \href{https://doi.org/10.1103/PhysRevD.90.065001}{\emph{Phys. Rev.}
  {\bfseries D90} (2014) 065001}
  [\href{https://arxiv.org/abs/1306.6575}{{\ttfamily 1306.6575}}].

\end{thebibliography}\endgroup

\end{document}